%% file: main.tex
\documentclass[runningheads,a4paper]{llncs}

\input{packages}
\input{macros_util}
\input{ggmacros}
\input{macros}

\title{
  A Dynamic Temporal Logic for Quality of Service in Choreographic Models
  \thanks{\tnxbehapi \tnxastra Carlos G.\ Lopez Pombo's research is
	 partly supported by Universidad de Buenos Aires by grant UBACyT
	 20020170100544BA and Agencia Nacional de Promoci\'on de la
	 Investigaci\'on, el Desarrollo Tecnol\'ogico y la Innovación
	 Científica through grant PICT-2019-2019-01793.
  	 \\[.5em]
	 The authors thank the anonymous reviewers for their constructive
	 comments.
	 \\[.5em]\
  }
}

\author{Carlos
  G. Lopez Pombo\thanks{On leave from CONICET-Universidad de Buenos Aires. Instituto de Ciencias de la Computaci\'{o}n, Buenos Aires, Argentina and Universidad de Buenos Aires, Facultad de Ciencias Exactas y Naturales, Departamento de Computaci\'{o}n, Buenos Aires, Argentina.
  }\inst{1}, Agust\'in E. Martinez Su\~n\'e\inst{2,3}, Emilio
  Tuosto\inst{4}
}
\institute{%
  Centro Interdisciplinario de Telecomunicaciones, Electrónica, Computación y Ciencia Aplicada - CITECCA, Universidad Nacional de Río Negro - Sede Andina, San Carlos de Bariloche, Argentina \\ \email{cglopezpombo@unrn.edu.ar} \and 
  CONICET-Universidad de Buenos Aires, Instituto de Ciencias de la Computaci\'{o}n, Buenos Aires, Argentina \and 
  Universidad de Buenos Aires, Facultad de Ciencias Exactas y Naturales, Departamento de Computaci\'{o}n, Buenos Aires, Argentina \\ \email{aemartinez@dc.uba.ar} \and
  Gran Sasso Science Institute, L'aquila, Italy \\   \email{emilio.tuosto@gssi.it}} \authorrunning{%
  Carlos G. Lopez Pombo, Agust\'in E. Martinez Su\~n\'e, Emilio Tuosto
}

\index{Lopez Pombo, Carlos G.}
\index{Martinez Su\~n\'e, Agust\'in E.}
\index{Tuosto, Emilio}

\begin{document}

\maketitle

\begin{abstract}

  %
  We propose a framework for expressing and analyzing the \emph{Quality of 
  Service} (QoS) of message-passing systems using a choreographic 
  model that consists of \emph{g-choreographies} and \emph{Communicating Finite 
  State machines} (\mbox{CFSMs}).
  %
  The following are our three main contributions: 
  \begin{inparaenum}[(I)]
  \item an extension of CFSMs with non-functional contracts to specify
	 quantitative constraints of local computations,
  \item a dynamic temporal logic capable of expressing QoS, properties
	 of systems relative to the g-choreography that specifies the
	 communication protocol,
\item the semi-decidability of our logic which enables a bounded
  model-checking approach to verify QoS property of communicating systems.
\end{inparaenum}
\end{abstract}

\section{Introduction}\label{sec:intro}
\input{intro}
\section{Related Work}\label{sec:rw}
\input{rw}
\section{Preliminaries}\label{sec:pre}
\input{preliminaries}
\section{Quality of Service of Communicating Systems}\label{sec:QoS-choreo}
\input{qoschoreo}
\section{$\QL$: A Dynamic Logic for QoS}\label{sec:logic}
\input{logic}
\section{A semidecision procedure for $\QL$}\label{sec:analysis}
\input{analysis}
\section{Conclusions}\label{sec:conclu}
\input{conclusions}


\input{biblio-long.tex}

\end{document}


%% file: packages.tex
\usepackage{amsfonts}
\usepackage{amsmath}
\usepackage[UKenglish]{babel}
\usepackage{paralist}
\usepackage[usenames,dvipsnames,svgnames,x11names,table]{xcolor}\definecolor{LavanderPosta}{HTML}{ebd9fc}
\usepackage{soul}
\usepackage{url}
\usepackage{graphicx}\graphicspath{ {images/} }
\usepackage{mathabx}
\usepackage[export]{adjustbox}
\usepackage[vlined,noend,linesnumbered]{algorithm2e}
\let\oldnl\nl
  \newcommand\nonl{%
    \renewcommand{\nl}{\let\nl\oldnl}}
\usepackage{mathtools}
\usepackage{stmaryrd}
\usepackage{centernot}
\usepackage{xspace}
\usepackage{todonotes}
\usepackage{braket}
\usepackage{paralist}

\usepackage{scalerel}

\usepackage{enumitem}
\newlist{todolist}{itemize}{2}
\setlist[todolist]{label=\huge$\square$}
\usepackage{pifont}
\usepackage{hyperref}
\hypersetup{
  colorlinks = true,
  linkcolor = purple,
  urlcolor  = purple,
  citecolor = teal,
  anchorcolor = teal
}
\usepackage[capitalise]{cleveref}


%% file: macros_util.tex
\usepackage{xargs}

\newcommand{\aCMq}{\aCM^{\mathit{QoS}}}
\newcommand{\grpOOp}{{\colorOp \lceil}}
\newcommand{\grpCOp}{{\colorOp \rfloor}}
\newcommand{\ggrp}[1][\aG]{\grpOOp #1 \grpCOp}

\newif\ifemi

\newcommand{\eMcomm}[2][check]{%
  \ifthenelse{\equal{#1}{new}}{{\color{red}#2}}{%
         \ifthenelse{\equal{#1}{changed}}{{\color{teal}{#2}}}{%
                \todo[color=orange!20]{\tiny eM: \color{NavyBlue}#1}%
                {\color{OliveGreen}{#2}}%
         }%
  }%
}

\newcommand{\tnxbehapi}[1][partly]{
  Research {#1} supported by the EU
  H2020 RISE programme under the Marie Skłodowska-Curie grant
  agreement No 778233.
}

\newcommand{\tnxastra}[1][partly]{%
  Research {#1} supported by the PRO3 MUR project
  Software Quality, and PNRR MUR project VITALITY (ECS00000041), Spoke
  2 ASTRA - Advanced Space Technologies and Research Alliance.
}

\DeclareGraphicsExtensions{%
  .png,.PNG,%
  .pdf,.PDF,%
  .jpg,.mps,.jpeg,.jbig2,.jb2,.JPG,.JPEG,.JBIG2,.JB2}

\usepackage[UKenglish]{babel}
  
\usepackage{fixme}
\fxusetheme{color}
\FXRegisterAuthor{eM}{aeM}{\color{orange} {\underline{eM}}}

\usepackage[normalem]{ulem} 

\usepackage{xifthen}        
\newcommand{\ifempty}[3]{%
  \ifthenelse{\isempty{#1}}{#2}{#3}%
}

\newcommand{\mkfun}[4][\colorFun]{
  \newcommand{#2}[1][#4]{
    {#1\textsf{#3}}
    \ifempty{##1}{}{
      ({##1})}
  }
}

\newcommand{\mkuop}[4][\colorFun]{
  \newcommand{#2}[1][#4]{
    {#1\textsf{#3}}
    \ifempty{##1}{}{
      \, {##1}}
  }
}

\newcommand{\hidden}[1]{}
\newcommand{\hide}[1]{}

\newcommand{\cf}[2]{
  \fontsize{#1}{#1}{\selectfont{#2}}
}
\ifemi
\usepackage{showlabels}

\newcommand{\emi}[2]{
  \marginpar{\fcolorbox{red}{shadecolor}{\cf{#1}{{#2}}}}
}
\newcommand{\emic}[2]{\par
  \fcolorbox{red}{shadecolor}{\parbox{\linewidth}{ 
      \color{gray}
      \begin{description}
      \item[{\color{blue} #2}]{\sf #1}
      \end{description}}}
}
\else
\newcommand{\emi}[2]{}
\newcommand{\emic}[2]{}{}
\fi

\newcommand{\sst}{\;\big|\;}

\newcommand{\conf}[1]{\ensuremath{\langle {#1} \rangle}}

\newcommand{\bnfdef}{\ ::=\ }
\newcommand{\bnfmid}{\;\ \big|\ \;}

\newcommand{\qqand}[1][and]{\qquad\text{#1}\qquad}

{\bfseries}{\rmfamily}
{\bfseries}{\rmfamily}

\newcommand{\quo}[1]{\lq\lq {#1}\rq\rq}
\def\finex{{\unskip\nobreak\hfil
\penalty50\hskip1em\null\nobreak\hfil$\diamond$
\parfillskip=0pt\finalhyphendemerits=0\endgraf}}

\definecolor{shadecolor}{rgb}{1,0.99,0.9}
\definecolor{bg}{rgb}{0.95,0.95,0.95}



%% file: ggmacros.tex
\usepackage{listings}
\usepackage{mfirstuc}
\usepackage{xstring}

\newcommand{\gsubs}[2]{^{#1} / _{#2}}
\newcommandx{\gsubst}[3][1=\aM,2=q,3=q',usedefault=@]{
  \left \{\gsubs{#3}{#2} \right \}#1
}
\newcommandx{\gsubsts}[5][1=\aM,2=q,3=q',4=q,5=q',usedefault=@]{
  \left \{\gsubs{#3}{#2}, \gsubs{#5}{#4} \right \}#1
}

\newcommandx{\wellpar}[2][1={\aG},2={\aG'},usedefault=@]{\mathit{wf}({#1}, {#2})}

\newif\ifcp
\cpfalse
\newcommand{\gname}[1][i]{\ifcp{\colorNode{\scriptstyle\textsf{#1}}}\else{}\fi}

\newif\ifguard
\guardfalse
\newcommand{\aguard}{\ifguard{\colorGuard \phi}\else{}\fi}

\def\colorGuard{\color{cyan}}
\def\colorPtp{\color{blue}}
\def\colorFun{\color{Navy}}

\def\colorOp{\color{OliveGreen}}
\def\colorNode{\color{LightCoral}}
\def\colorR{\color{OliveGreen}}
\def\colorE{\color{orange}}

\def\colorMsg{\color{BrickRed}}

\newcommand{\fillcolor}{orange!5}

\newcommand{\toolidcol}{red!10!blue!90}
\newcommand{\toolid}[1]{\textcolor{\toolidcol}{\textbf{\textsf{#1}}}}

\newcommand{\chorgramsite}{\url{https://bitbucket.org/eMgssi/chorgram/src/master/}}
\newcommand{\chorgram}{\toolid{ChorGram}}

\newcommand{\msg}[1][m]{\mathsf{\colorMsg{#1}}}
\newcommand{\msgset}{\mathcal{\colorMsg{M}}}

\newcommand{\chset}{\mathcal{C}}

\newcommand{\lset}{\mathcal{L}}

\newcommand{\ptp}[1][A]{\ensuremath{\mathsf{\colorPtp{\capitalisewords{#1}}}}}

\newcommand{\p}{\ptp}
\newcommand{\q}{{\ptp[B]}}
\newcommand{\s}{{\ptp[S]}}

\newcommand{\sndint}[1][{\gint[]}]{\mathrm{snd}(#1)}
\newcommand{\rcvint}[1][{\gint[]}]{\mathrm{rcv}(#1)}
\newcommandx{\ggcommon}[3][1=\ptp,2={\aH},3={\aH'},usedefault=@]{f_{#1}}
\newcommandx{\opair}[2][1={\ae},2={\ae'},usedefault=@]{\conf{{#1},{#2}}}
\newcommandx{\hopair}[2][1={\aE},2={\aE'},usedefault=@]{\llparenthesis\, {#1},{#2}\, \rrparenthesis}
\newcommandx{\wf}[2][1={\aG},2={\aG'},usedefault=@]{wf({#1}, {#2})}
\newcommandx{\wb}[2][1={\aG},2={\aG'},usedefault=@]{wb({#1}, {#2})}
\newcommandx{\ws}[2][1={\aG},2={\aG'},usedefault=@]{ws({#1}, {#2})}

\newcommandx{\widx}[2][1={\aW},2={i},usedefault=@]{{#1}[{#2}]}
\newcommandx{\outop}[2][1=\gname,2={}]{{\colorOp{!}}^{{#1}{#2}}}
\newcommandx{\inop}[2][1=\gname,2={}]{{\colorOp{?}}^{{#1}{#2}}}
\newcommandx{\aout}[4][1=a,2=b,3=m,4={},usedefault=@]{
  \achan[#1][#2] \outop[{}] {\msg[#3]} {#4}
}
\newcommandx{\ain}[4][1={\p},2={\q},3=m,4={},usedefault=@]{
  \achan[#1][#2] \inop[{}] {\msg[{#3}]}{#4}
}
\newcommandx{\adep}[1][1={}]{
  \conf{ \aout[@][@][@][@][{#1}], \ain[@][@][@][@][{#1}]}
}

\newcommandx{\hproj}[2][1=\aH, 2=\ptp, usedefault=@]{
  \ifempty{#1}{}{{#1}}\ifempty{#2}{}{{^{\scriptscriptstyle @{#2}}}}
}
\newcommandx{\eproj}[2][1=\aE,2=A, usedefault=@]{
  {{#1}}\ifempty{#2}{}{{^{\scriptscriptstyle @{{\ptp[#2]}}}}}
}

\tikzset{
  component/.style={
    draw,
    fill = white,
    minimum width = 1.5cm,
    minimum height = .5cm,
    drop shadow
  }
}
\tikzset{
  file/.style={
    thin,
    fill = blue!5,
    font = \tt\scriptsize,
    text width = .8cm,
    minimum width = 1.0cm,
    minimum height = .5cm,
    drop shadow
  }
}

\tikzset{
  dataflow/.style={
    thick,
    draw, ->, >=latex,
    dashed,
    OliveGreen
  }
}

\tikzset{
  pipeline/.style={
    thick,
    draw, ->, >=latex,
    double,
    red
  }
}

\tikzset{
  pomsetcloud/.style={
    cloud,
	 cloud puffs=20,
	 cloud ignores aspect,
	 minimum height=.1cm,
	 minimum width=2cm,
	 fill=blue!10,
	 opacity=.5,
	 draw
  }
}

\newcommand{\apom}{r}

\newcommand{\aR}[1][R]{{\colorR{#1}}}

\newcommand{\aConf}{s}

\newcommandx{\detM}[1][1=\aCM,usedefault=@]{\Delta({#1})}

\tikzset{
  cnode/.style={
    shape=circle,
    minimum size = 0mm,
    inner sep = 1pt,
    font=\tiny,
    draw
  },
  carrow/.style={
    ->,
    shorten >=1pt,
    >=stealth',
    auto,
    font=\scriptsize,
    draw,
    sloped
  }
}

\newcommandx{\choranimation}[4][1=1,2=1,3=notempty,4=.35,usedefault=@]{
  \begin{overlayarea}{#1\linewidth}{#2\textheight}
    \begin{tikzpicture}[
      node distance=1cm and 2cm,
      scale=#4,
      every node/.style={transform shape},
      font=\large
      ]
      \node [choreo, align=center] (global){Choreography $\aG$ \\ global viewpoint};
      \node [below= of global] (fake) {};
      \node [choreo, align=center, local, below =of fake] (typei) {$\aCM_1$ \\ Local viewpoint$_1$};
      \node [choreo, align=center, local, left =of typei] (type1) {$\aCM_i$ \\ Local viewpoint$_1$};
      \node [choreo, align=center, local, right=of typei] (typen) {$\aCM_n$ \\ Local viewpoint$_n$};
      \node<+-> (synctxt) at (9,0)  {\textcolor{Navy}{\bf Synchrony}};
      \node<.-> (asynctxt) at (9,-4) {\textcolor{Navy}{\bf Asynchrony}};
      \node<.-> [below=of type1] (fake1) {};
      \node<.-> [below=of typei] (fakei) {};
      \node<.-> [below=of typen] (faken) {};
      \path<+-> [bigar] (global) edge[sloped,above] node {\color{OliveGreen}Project} (type1);
      \path<.-> [bigar] (global) edge[sloped,above] node {\color{OliveGreen}Project} (typei);
      \path<.-> [bigar] (global) edge[sloped,above] node {\color{OliveGreen}Project} (typen);
      \path<.-> [elli] (type1) -- (typei);
      \path<.-> [elli] (typei) -- (typen);
      \node<+-> [process, below=of fake1] (proc1) {Component$_1$};
      \node<.-> [process, below=of fakei] (proci) {Component$_i$};
      \node<.-> [process, below=of faken] (procn) {Component$_n$};
      \path<.-> [bigar,->,dashed,gray] (type1) edge[sloped,above] node {Validate} (proc1);
      \path<.-> [bigar,->,dashed,gray] (typei) edge[sloped,above] node {Validate} (proci);
      \path<.-> [bigar,->,dashed,gray] (typen) edge[sloped,above] node {Validate} (procn);
      \path<.-> [elli] (proc1) -- (proci);
      \path<.-> [elli] (proci) -- (procn);
		\ifempty{#3}{}{
      \node<+-> [process, right=of procn,xshift=4cm] (evolve1) {Component'$_1$};
      \node<.-> [process, right=of evolve1] (evolvei) {Component'$_i$};
      \node<.-> [process, right=of evolvei] (evolven) {Component'$_n$};
      \path<.-> [bigar,blue,dotted] (procn) edge node [above] {evolve/replace/compose} (evolve1);      
      \path<.-> [elli] (evolve1) -- (evolvei);
      \path<.-> [elli] (evolvei) -- (evolven);
      \node<+-> [choreo, align=center, local, above=of evolve1, yshift=1.5cm] (t11) {New $\aCM'_1$ \\ Local viewpoint$_1$};
      \node<.-> [choreo, align=center, local, above=of evolvei, yshift=1.5cm] (t1i) {New $\aCM'_i$ \\ Local viewpoint$_i$};
      \node<.-> [choreo, align=center, local, above=of evolven, yshift=1.5cm] (t1n) {New $\aCM'_n$ \\ Local viewpoint$_n$};
      \path<.-> [elli] (t11) -- (t1i);
      \path<.-> [elli] (t1i) -- (t1n);
      \path<.-> [bigar,->,dashed,gray] (evolve1) edge[sloped,above] node {Extract} (t11);
      \path<.-> [bigar,->,dashed,gray] (evolvei) edge[sloped,above] node {Extract} (t1i);
      \path<.-> [bigar,->,dashed,gray] (evolven) edge[sloped,above] node {Extract} (t1n);
      \node<+> [above=of t1i, yshift=1.5cm] (qm) {\Huge \textcolor{red}{ ??? }};
      \node<.-> [above=of t1i] (dummy) {};
      \node<+-> [choreo,align=center,above=of dummy,scale=.85] (global') {New choreography $\aG'$ \\ global viewpoint};
      \path<.-> [bigar,-] (t11) -- (dummy);
      \path<.-> [bigar,->] (t1i) edge node[right,xshift=1em] {\color{OliveGreen}Synthesise} (global');
      \path<.-> [bigar,-] (t1n) -- (dummy);
		}
    \end{tikzpicture}
  \end{overlayarea}
}

\newcommandx{\cm}[2][1=\ptp, 2=\aM]{{#2}_{#1}}
\newcommandx{\achan}[2][1=A,2=B,usedefault=@]{{\ptp[#1]}{\,}{\ptp[#2]}}
\newcommand{\ptpset}{\mathcal{\colorPtp{P}}}
\newcommand{\PSet}{\ptpset}

\newcommand{\oact}{\outop[]}
\newcommand{\iact}{\inop[]}
\newcommand{\tset}{\to}

\newcommand{\csconf}[2]{\conf{\vec{#1} \ ; \ \vec{#2}}}
\newcommand{\chanset}{\chset}
\newcommand{\TRANSS}[1]{{\xRightarrow{\raisebox{-.3ex}[0pt][0pt]{$\scriptstyle #1$} }}}

\newcommand{\trans}[2][{}]{\,\xrightarrow{#2}_{#1}\,}

\newcommandx{\acfsmout}[3][1=A,2=B,3=m,usedefault=@]{\achan[{#1}][{#2}] \oact {\msg[{#3}]}}
\newcommandx{\acfsmin}[3][1=A,2=B,3=m,usedefault=@]{\achan[{#1}][{#2}] \iact {\msg[{#3}]}}
\newcommandx{\fsaout}[2][1={\p},2={},usedefault=@]{
  \ptp[#1] \ \outop[]\ \msg[{#2}]
}
\newcommandx{\fsain}[2][1={\p},2={},usedefault=@]{
  \ptp[#1] \ \inop[]\ \msg[{#2}]
}

\makeatletter
\newcommand{\linenumfontsize}{\@setfontsize{\linenumfontsize}{3pt}{3pt}}
\makeatother
\lstdefinelanguage{sys}{
	commentstyle=\color{Gray},
	morecomment=[s]{[}{]},
	keywords=[0]{system,of,do,end},	keywordstyle=\color{orange}\bfseries,
}

\lstdefinelanguage{sgg}{
  commentstyle=\color{Gray},
  morecomment=[l]{..},
  morecomment=[s]{[}{]},
  keywords=[0]{repeat,branch,sel},
  keywordstyle=\color{orange}\bfseries,
  morekeywords=[1]{*,\+,|,->},
  literate={->}{$\colorOp \xrightarrow$}1 {|}{$\gparop$}1 {;}{$\gseqop$}1 {+}{$\gchoop$}1 {\{}{{\textcolor{Navy}{\{}}}1 {\}}{{\textcolor{Navy}{\}}}}1
}

\newcommand{\aG}{\mathsf{G}}

\newcommand{\gseqop}{{\colorOp ;}\,}
\newcommand{\gparop}{{\colorOp \ |\ }}
\newcommand{\gchoop}{{\colorOp \ +\ }}
\newcommand{\grecop}{{\colorOp *}}

\newcommandx{\nmerge}[2][1={i},2={},usedefault=@]{
  \ifempty{#2}{
    \ifempty{#1}{\mu}{\gname[-{#1}]}
  }{-{#2}}
}

\mkfun{\esbj}{sbj}{\ae}

\makeatletter%
\@ifclassloaded{exam-paper}%
  {}%
  {\makeatletter%
    \@ifclassloaded{test}%
    {}%
    \makeatother%
  }
\makeatother%

\newcommandx{\gnode}[2][1=i,2=\gint,usedefault=@]{
  \ifcp{
    \ifempty{#1}{#2}{\gname[#1].\big({#2}\big)}
  }
  \else
  {#2}
  \fi
}

\newcommand{\gempty}{\mathbf{0}}
\newcommandx{\refgint}[3][1=A,2=\msg,3=B,usedefault=@]{
  \ptp[#1] {\colorOp \xdashrightarrow[{}]{\msg[{#2}]}}{
    
	 \docsvlist{#3}
  }
}
\newcommandx{\gint}[4][1=i,2=A,3=m,4=B,usedefault=@]{
  \ptp[#2] {\colorOp \xrightarrow{\scriptstyle \gname[#1]}} \ptp[#4] {\colorOp \colon} {\msg[{#3}]}
}
\newcommandx{\gout}[4][1=\gname,2=\ptp,3=m,4={\ptp[C]},usedefault=@]{
  \achan[{#2}][{#4}] {\colorOp {\colorOp{!}}} {\msg[{#3}]}
}
\newcommandx{\gin}[4][1=\gname,2=\ptp,3=m,4={\ptp[C]},usedefault=@]{
  \achan[{#2}][{#4}] {\colorOp {\colorOp{?}}} {\msg[{#3}]}
}
\newcommandx{\gseq}[3][1=\gname,2={\aG},3={\aG'},usedefault=@]{
  \def\ggraph{{#2} \gseqop {#3}}
  \ggraph
}
\newcommand{\ginfix}[4]{
  \def\ggraph{{#2} {#4} {#3}}
  \gnode[#1][\ggraph]
}
\newcommandx{\gpar}[3][1=i,2={\aG},3={\aG'},usedefault=@]{
  \ginfix{#1}{#2}{#3}{\gparop}
}
\newcommandx{\gcho}[3][1=i,2={\aG},3={\aG'},usedefault=@]{
  \ginfix{#1}{#2}{#3}{\gchoop}
}
\newcommandx{\gchov}[3][1=\gname,2={\aG},3={\aG'},usedefault=@]{
  \def\ggraph{\left(
  \begin{array}l
    \ifempty{#1}{{#2} \\ \gchoop \\ {#3}}{\!\!{#2} \\ \gchoop \\ {#3}}
  \end{array}\right)}
  \ifcp\gnode[{$#1$}][\ggraph] \else \ggraph \fi
}
\newcommandx{\grec}[2][1={},2={\aG},usedefault=@]{
  \def\ggraph{{#2}^\grecop}
  \ifempty{#1}{\ggraph}{\gname[{$#1$}][\ggraph]}
}	

\newcommand{\getcentroid}[2]{
    \coordinate (tmpgatecoord) at (0,0);
    \foreach \n [count=\i] in {#1}{
      \path (\n);
      \coordinate (tmpgatecoord) at ($(tmpgatecoord) + (\n)$);
      \coordinate (#2) at ($1/\i*(tmpgatecoord)$);
    }
}

\tikzset{
  hgsem/.style={
    draw,
    node distance=2cm and 1cm,
    transform shape,
    smooth,
    every node/.style = {font=\sffamily\bfseries}
  }
}

\tikzset{
  hgstyle/.style={
    src color={#1},
    tgt color={#1},
    centroid color={#1},
    centroid label={#1},
    centroid name={#1},
    centroid radius={#1},
    centroid ratio={#1},
    xoffset={#1},
    yoffset={#1},
    xsrcoffset={#1},
    ysrcoffset={#1},
    xtgtoffset={#1},
    ytgtoffset={#1},
    font={#1},
    centroid angle={#1},
    centroid tolerance={#1}
  },
  src color/.store in = \hgsrccol,
  tgt color/.store in = \hgtgtcol,
  centroid color/.store in =\hgfillcolor,
  centroid label/.store in =\hglabel,
  centroid name/.store in =\hgname,
  centroid radius/.store in = \hgradius,
  centroid ratio/.store in = \hgratio,
  xoffset/.store in =\hgxoffset,
  yoffset/.store in =\hgyoffset,
  xsrcoffset/.store in =\hgxsrcoffset,
  ysrcoffset/.store in =\hgysrcoffset,
  xtgtoffset/.store in =\hgxtgtoffset,
  ytgtoffset/.store in =\hgytgtoffset,
  centroid angle/.store in =\hgangle,
  centroid tolerance/.store in =\hgtolerance,
  src color = black,
  tgt color = black,
  centroid color = orange!40,
  centroid label={},
  centroid name={dummycentroid},
  centroid radius = .7pt,
  centroid ratio = .35,
  xoffset = 0,
  yoffset = 0,
  xsrcoffset = 0,
  ysrcoffset = 0,
  xtgtoffset = 0,
  ytgtoffset = 0,
  font=\sffamily\scriptsize,
  centroid angle=0,
  centroid tolerance=10pt
}

\newcommandx{\mkhg}[5][1={},4={},5={},usedefault=@]{
  \begingroup
  \tikzset{#1}
  \StrCount{#2,}{,}[\l] 
  \StrCount{#3,}{,}[\m] 
  \ifthenelse{\l = 1 \AND \m = 1}{
    \ifempty{#4}{
      \ifempty{#5}{
        \path[hgsem, ->, >=stealth', shorten >=1pt] (#2) -- (#3);
      }{
        \path[hgsem, ->, >=stealth', shorten >=1pt] (#2) #5 (#3);
      }
    }{
      \ifempty{#5}{
        \path[hgsem, ->, >=stealth', shorten >=1pt, #4] (#2) -- (#3);
      }{
        \path[hgsem, ->, >=stealth', shorten >=1pt, #4] (#2) #5 (#3);
      }
    }
  }{
    \coordinate (srcoffset) at (\hgxsrcoffset,\hgysrcoffset);
    \coordinate (tgtoffset) at (\hgxtgtoffset,\hgytgtoffset);
    \getcentroid{#2}{srccentroid};
    \getcentroid{#3}{tgtcentroid};
    \node[label={left:\hglabel}] (\hgname) at ($(srccentroid)!{1-\hgratio}!\hgangle:(tgtcentroid) + (\hgxoffset,\hgyoffset)$) {};
    \pgfgetlastxy \xc \yc;
    \pgfmathtruncatemacro{\xcontrol}{\xc};
    \pgfmathtruncatemacro{\ycontrol}{\yc};
    \foreach \n in {#2}{
      \path (\n);
      \pgfgetlastxy \xntmp \yntmp;
      \pgfmathtruncatemacro{\xn}{\xntmp};
      \pgfmathtruncatemacro{\yn}{\yntmp};
      \pgfmathsetmacro\xtmpdiff{abs(\xn - \xcontrol + \hgxsrcoffset)};
      \pgfmathsetmacro\ytmpdiff{abs(\yn - \ycontrol + \hgytgtoffset)};
      \ifdim \xtmpdiff pt > \hgtolerance
      \ifempty{#4}{
        \path[hgsem, \hgsrccol] (\n) .. controls ($(srccentroid.center) + (srcoffset)$) .. (\hgname.center);
      }{
        \path[hgsem, \hgsrccol] (\n) .. controls ($(srccentroid.center) + (srcoffset)$) .. (\hgname.center);
      }
      \else
      \ifempty{#4}{
        \path[hgsem, \hgsrccol] (\n) -- (\hgname.center);
      }{
        \path[hgsem, \hgsrccol, #4] (\n) -- (\hgname.center);
      }
      \fi
    }
    \foreach \n in {#3}{
      \path (\n);
      \pgfgetlastxy \xntmp \yntmp;
      \pgfmathtruncatemacro{\xn}{\xntmp};
      \pgfmathtruncatemacro{\yn}{\yntmp};
      \pgfmathsetmacro\xtmpdiff{abs(\xn - \xcontrol)};
      \pgfmathsetmacro\ytmpdiff{abs(\yn - \ycontrol)};
      \ifdim \xtmpdiff pt > \hgtolerance
      \ifempty{#4}{
        \path[hgsem, ->, >=stealth', shorten >=1pt, \hgtgtcol] (\hgname.center) .. controls (tgtcentroid.center) and ($(tgtcentroid.center) + (tgtoffset)$) .. (\n);
      }{
        \path[hgsem, ->, >=stealth', shorten >=1pt, \hgtgtcol,#4] (\hgname.center) .. controls (tgtcentroid.center) and ($(tgtcentroid.center) + (tgtoffset)$) .. (\n);
      }
      \else
      \ifempty{#4}{
        \path[hgsem, ->, >=stealth', shorten >=1pt, \hgtgtcol] (\hgname.center) --  (\n);
      }{
        \path[hgsem, ->, >=stealth', shorten >=1pt, \hgtgtcol] (\hgname.center) --  (\n);
      }
      \fi
    }
    \fill[\hgfillcolor] (\hgname) circle [radius=\hgradius];
  }
  \endgroup
}

\newcommandx{\hgordeq}[1][1={\aH},usedefault=@]{\sqsubseteq_{#1}}
\newcommandx{\gintsem}[4][4=.5]{
  \tikz[hgsem,scale=#4,every node/.style={font=\scriptsize}]{
    \node (out) {$\aout[{#1}][{#2}][][{#3}]$};
    \node[below = 20pt of out] (in) {$\ain[{#1}][{#2}][][{#3}]$};
    \mkhg{out}{in};
  }
}

\newcommandx{\gsem}[2][1={\aG},2={},usedefault=@]{\left\llbracket {#1} \right\rrbracket_{#2}}

\newcommandx{\rbot}{\text{undef}}
\newcommandx{\rtrs}[1][1={\aH},usedefault=@]{{#1}^{\star}}
\newcommandx{\gord}[1][1={\aG},usedefault=@]{\leq_{#1}}
\newcommandx{\gordeq}[1][1={\aG},usedefault=@]{\leq_{#1}}
\mkfun{\cause}{cs}{}
\mkfun{\effect}{ef}{}

\newcommandx{\aW}{w}
\newcommandx{\rlang}{\mathcal{L}}

\newcommand{\gfun}[1]{\ensuremath{\mathsf{\colorFun #1}}}
\mkfun{\eact}{\gfun{act}}{}
\mkfun{\enode}{\gfun{cp}}{}

\mkuop{\rmax}{\gfun{max}}{\aH}
\mkuop{\rmin}{\gfun{min}}{\aH}
\mkuop{\rMAX}{\gfun{lst}}{\aH}
\mkuop{\rMIN}{\gfun{fst}}{\aH}

\newcommandx{\rseq}[2][1=\aG,2={\aG'},usedefault=@]{\gfun{seq}({#1},{#2})}
\newcommandx{\rpar}[2][1=\aG,2={\aG'},usedefault=@]{\gfun{par}({#1},{#2})}

\newcommandx{\gproj}[2][1=\aG,2=\ptp]{{#1}\downarrow_{#2}}
\newcommandx{\cinit}[1][1={\aQzero},usedefault=@]{{#1}}
\newcommandx{\cfinal}[1][1={q_e},usedefault=@]{{#1}}

\newcommandx{\geproj}[4][1=\aG,2=\ptp,3=\cinit,4=\cfinal,usedefault=@]{
  {#1}\downarrow_{#2}^{{#3},{#4}}
}

\newcommand*{\StrikeThruDistance}{0.15cm}%
\tikzset{strike thru arrow/.style={
    decoration={markings, mark=at position 0.5 with {
        \draw [blue, thick,-] 
            ++ (-\StrikeThruDistance,-\StrikeThruDistance) 
            -- ( \StrikeThruDistance, \StrikeThruDistance);}
    },
    postaction={decorate},
}}

\newcommandx{\ich}[1][1={\aG},usedefault=@]{{#1}^{\oplus}}
\newcommandx{\ichedges}[2][1={\aG},2={\gname},usedefault=@]{{#1}^{\oplus}({#2})}
\newcommandx{\parts}[1]{2^{#1}}
\newcommandx{\actch}{c}
\newcommandx{\soundactch}[2][1={\aG},2={\actch},usedefault=@]{{#1} \,\circledR\, {#2}}
\newcommandx{\rOnActch}[2][1={\aG},2={\actch},usedefault=@]{{#1} \setminus {#2}}
\newcommandx{\rOnActchClean}[2][1={\aG},2={\actch},usedefault=@]{{#1} \circledR {#2}}
\newcommandx{\rAllEvents}[1][1={\aG},usedefault=@]{\mathit{dom}(#1)}

\newcommand{\AV}{\mathcal{V}}
\newcommand{\aH}{H}

\newcommandx{\hgvertex}[2][1=\al,2=\gname,usedefault=@]{{#1}_{\textcolor{red}{[{#2}]}}}
\newcommand{\aE}{{\colorE E}}
\renewcommand{\ae}[1][e]{{\colorE{#1}}}

\newcommand{\al}[1][\ell]{{\colorE{#1}}}
\newcommandx{\hyedge}[1]{\{#1\}}

\newcommandx{\rdiv}[2][1=\gcho,2=\ptp,usedefault=@]{
  \gfun{div}_{#2}(#1)
}

\newcommandx{\rrdiv}[5][1={\aG},2={\aG'},3={\AV},4={,\AV'},5=\ptp,usedefault=@]{
  \gfun{div}^{#3#4}_{#5}(#1,#2)
}
\newcommandx{\pdiv}[3][1={\apom_1},2={\apom_2},3={\apom},usedefault=@]{
  \gfun{div}_{#3}(#1,#2)
}
\newcommandx{\pfork}[3][1={\apom_1},2={\apom_2},3={\apom},usedefault=@]{
  \gfun{fork}_{#3}(#1,#2)
}

\tikzset{
  pomset/.style={
    node distance = .6cm and .6cm,
    scale = .7,
    transform shape,
    smooth
  }
}

\newcommandx{\mkint}[6][3=i,4=\p,5=\msg,6=\q,usedefault=@]{
  \node[bblock, #1] (#2) {$\gint[#3][#4][#5][#6]$};
}

\newcommand{\mkseq}[2]{\path[line] (#1) -- (#2);}

\newcommand{\mknseq}[1]{
  \StrCount{#1}{,}[\l] 
  \StrBefore{#1}{,}[\myhead]
  \StrBehind{#1}{,}[\mytail]
  \StrBefore{\mytail}{,}[\sndel]
  \ifnum \l > 1 {
    \mkseq{\myhead}{\sndel};
    \mknseq{\mytail}
  }
  \else{\ifnum \l > 0{
      \mkseq{\myhead}{\mytail};
    }
    \else{}
    \fi
  }
  \fi
}

\newcommandx{\mkgateblock}[6][6=yellow!10,usedefault=@]{
  \path(#2);
  \pgfgetlastxy{\xgate}{\ygate};
  \pgfmathtruncatemacro{\xgateround}{\xgate};
  \StrCount{#3,}{,}[\l] 
  \ifnum \l < 2 {\errmessage{#3 argument should be a comma-separated list of lenght >= 2}}
  \else{
    \foreach \n in {#3}{
      \path (\n);
      \pgfgetlastxy{\xnode}{\ynode};
      \pgfmathtruncatemacro{\xnround}{\xnode};
      \pgfmathsetmacro\tmpdiff{abs(\xnround - \xgateround)}
      \ifdim \tmpdiff pt > 1 pt \path[line] (#2) -| (\n);
      \else
        \path[line] (#2) -- (\n);
      \fi
    }
  }
  \fi
  \StrCount{#4,}{,}[\l] 
  \ifnum \l < 2 {\errmessage{#4 argument should be a comma-separated list of lenght >= 2}}
  \else{
    \foreach \n in {#4}{
      \path (\n);
      \pgfgetlastxy{\xnode}{\ynode};
      \pgfmathtruncatemacro{\xnround}{\xnode};
      \pgfmathsetmacro\tmpdiff{abs(\xnround - \xgateround)}
      \ifdim \tmpdiff pt > 1 pt \path[line] (\n) |- (#5);
      \else
        \path[line] (\n) -- (#5);
      \fi
    }
  }
  \fi
  \node[#1] at (#2) {};
  \node[#1] at (#5) {};
  {
    \begin{pgfonlayer}{background}
      \path[fill=#6,rounded corners]
      (current bounding box.south west) rectangle
      (current bounding box.north east);
    \end{pgfonlayer}
  }
}

\newcommandx{\mkbranchblock}[5][5=@]{
  \mkgateblock{ogate}{#1}{#2}{#3}{#4}[#5]
}

\newcommandx{\mkforkblock}[5][5=@]{
  \mkgateblock{agate}{#1}{#2}{#3}{#4}[#5]
}

\newcommandx{\mkgraph}[3][1=.5cm, usedefault=@]{
  \node[source,above = #1 of {#2}] (src#2) {};
  \node[sink,below  = #1 of {#3}] (sink#3) {};
  \path[line] (src#2) -- (#2);
  \path[line] (#3) -- (sink#3);
}

\newcommandx{\mkloop}[5][1=.5, 2=1.5, 5=\aguard, usedefault=@]{
  \node[lgate,above = #1 of {#3}] (entry#3) {};
  \pgfgetlastxy \xentry \yentry;
  \pgfmathtruncatemacro{\xentryrounded}{\xentry};
  \node[below = #1 of {#4}, label = {above right:{$#5$}},yshift=-1em] (dummy) {};
  \node[lgate,below  = #1 of {#4}] (exit#4) {};
  \pgfgetlastxy \xexit \yexit;
  \pgfmathtruncatemacro{\xexitrounded}{\xexit};
  \path[line] (entry#3) -- (#3);
  \path[line] (#4) -- (exit#4);
  \pgfmathsetmacro\tmpdiff{abs(\xentryrounded - \xexitrounded)}
  \path[line, color=teal] (exit#4) -| ($(exit#4)+(\tmpdiff,0)+(#2,0)$) |- (entry#3);
}

\newcommandx{\mkfork}[4][2=gatenode,3=i,4=.6,usedefault=@]{
  \mkgatebegin{#1}[{\gname[{#3}]}][agate][#4]{#2}
}

\newcommandx{\mkbranch}[4][2=gatenode,3=i,4=.6,usedefault=@]{
  \mkgatebegin{#1}[{\gname[{#3}]}][ogate][#4]{#2}
}

\newcommandx{\mkgatebegin}[5][2={},3=ogate,4=.5,usedefault=@]{
  \coordinate (gatecord) at (0,0);
  \coordinate (xmax) at (0,0);
  \coordinate (xmin) at (0,0);
  \pgfgetlastxy \xmin \xmax;
  \foreach \n [count=\i] in {#1}{
    \pgfgetlastxy \xc \yc;
    \path (\n);
    \pgfgetlastxy \xn \yn;
    \ifnum \i = 1
      \coordinate (xmin) at (\xn,0);
      \coordinate (xmax) at (\xn,0);
      \coordinate (max) at (0,\yn);
    \else
      \ifdim \xn < \xmin
        \coordinate (xmin) at (\xn,0);
      \fi
      \ifdim \xn > \xmax
        \coordinate (xmax) at (\xn,0);
      \fi
      \ifdim \yn < \yc
        \coordinate (max) at (0,\yc);
      \else
        \coordinate (max) at (0,\yn);
      \fi
    \fi
  }
  \coordinate (gatecord) at ($(xmin)!.5!(xmax) + (max) + (0,#4) + (max)$);
  \node[#3,label={below:$#2$}] (#5) at (gatecord) {};
  \pgfgetlastxy{\xgate}{\ygate};
  \pgfmathtruncatemacro{\xgateround}{\xgate};
  \StrCount{#1,}{,}[\l] 
  \ifnum \l < 2 {\errmessage{#1 argument should be a comma-separated list of lenght >= 2}}
  \else{
    \foreach \n in {#1}{
      \path (\n);
      \pgfgetlastxy{\xnode}{\ynode};
      \pgfmathtruncatemacro{\xnround}{\xnode};
      \pgfmathsetmacro\tmpdiff{abs(\xnround - \xgateround)}
      \ifdim \tmpdiff pt > 1 pt \path[line] (#5) -| (\n);
      \else
        \path[line] (#5) -- (\n);
      \fi
    }
  }
  \fi
}

\newcommandx{\mkgatebeginold}[5][2={},3=ogate,4=.5,usedefault=@]{
  \coordinate (gatecord) at (0,0);
  \foreach \n [count=\i] in {#1}{
    \pgfgetlastxy \xc \yc;
    \path (\n);
    \pgfgetlastxy \xn \yn;
    \coordinate (gatecord) at ($(gatecord) + (\xn,0)$);
    \coordinate (gatecord) at ($1/\i*(gatecord)$);
    \ifdim \yn < \yc
    \node (max) at (0,\yc) {};
    \else
    \node (max) at (0,\yn) {};
    \fi
  }
  \coordinate (gatecord) at ($(gatecord) + (0,#4) + (max)$);
  \node[#3,label={below:$#2$}] (#5) at (gatecord) {};
  \pgfgetlastxy{\xgate}{\ygate};
  \pgfmathtruncatemacro{\xgateround}{\xgate};
  \StrCount{#1,}{,}[\l] 
  \ifnum \l < 2 {\errmessage{#1 argument should be a comma-separated list of lenght >= 2}}
  \else{
    \foreach \n in {#1}{
      \path (\n);
      \pgfgetlastxy{\xnode}{\ynode};
      \pgfmathtruncatemacro{\xnround}{\xnode};
      \pgfmathsetmacro\tmpdiff{abs(\xnround - \xgateround)}
      \ifdim \tmpdiff pt > 1 pt \path[line] (#5) -| (\n);
      \else
        \path[line] (#5) -- (\n);
      \fi
    }
  }
  \fi
}

\newcommandx{\mkmerge}[4][2=gatenode,3=i,4=.5,usedefault=@]{
  \mkgateend{#1}[{\ifempty{#3}{}{\nmerge[#3]}}][ogate][#4]{#2}
}

\newcommandx{\mkjoin}[4][2=gatenode,3=i,4=.5,usedefault=@]{\mkgateend{#1}[{\ifempty{#3}{}{\nmerge[#3]}}][agate][#4]{#2}}

\newcommandx{\mkgateend}[5][2={},3=ogate,4=.5,usedefault=@]{
  \coordinate (gatecord) at (0,0);
  \coordinate (xmax) at (0,0);
  \coordinate (xmin) at (0,0);
  \pgfgetlastxy \xmin \xmax;
  \foreach \n [count=\i] in {#1}{
    \pgfgetlastxy \xc \yc;
    \path (\n);
    \pgfgetlastxy \xn \yn;
    \ifnum \i = 1
      \coordinate (xmin) at (\xn,0);
      \coordinate (xmax) at (\xn,0);
      \coordinate (ymin) at (0,\yn);
    \else
      \ifdim \xn < \xmin
        \coordinate (xmin) at (\xn,0);
      \fi
      \ifdim \xn > \xmax
        \coordinate (xmax) at (\xn,0);
      \fi
      \ifdim \yn > \yc
        \coordinate (ymin) at (0,\yc);
      \else
        \coordinate (ymin) at (0,\yn);
      \fi
    \fi
  }
  \coordinate (gatecord) at ($(xmin)!.5!(xmax) + (ymin)$);
  \node[#3,label={above:$#2$}] (#5) at ($(gatecord) - (0,{#4})$) {};
  \pgfgetlastxy{\xgate}{\ygate};
  \pgfmathtruncatemacro{\xgateround}{\xgate};
  \StrCount{#1,}{,}[\l] 
  \ifnum \l < 2 {\errmessage{#1 argument should be a comma-separated list of lenght >= 2}}
  \else{
    \foreach \n in {#1}{
      \path (\n);
      \pgfgetlastxy{\xnode}{\ynode};
      \pgfmathtruncatemacro{\xnround}{\xnode};
      \pgfmathsetmacro\tmpdiff{abs(\xnround - \xgateround)}
      \ifdim \tmpdiff pt > 1 pt \path[line] (\n) |- (#5);
      \else
        \path[line] (\n) -- (#5);
      \fi
    }
  }
  \fi
}

\newcommandx{\mkgateendold}[5][2={},3=ogate,4=.5,usedefault=@]{
  \coordinate (gatecord) at (0,0);
  \coordinate (xmax) at (0,0);
  \coordinate (xmin) at (0,0);
  \pgfgetlastxy \xmin \xmax;
  \foreach \n [count=\i] in {#1}{
    \pgfgetlastxy \xc \yc;
    \path (\n);
    \pgfgetlastxy \xn \yn;
    \ifdim \xn < \xmin
    \coordinate (xmin) at (\xn,0);
    \fi
    \ifdim \xn > \xmax
    \coordinate (xmax) at (\xn,0);
    \fi
    \ifdim \yn > \yc
    \coordinate (ymin) at (0,\yc);
    \else
    \coordinate (ymin) at (0,\yn);
    \fi
    \coordinate (gatecord) at ($(xmin)!.5!(xmax) + (ymin)$);
  }
  \node[#3,label={above:$#2$}] (#5) at ($(gatecord) - (0,{#4})$) {};
  \pgfgetlastxy{\xgate}{\ygate};
  \pgfmathtruncatemacro{\xgateround}{\xgate};
  \StrCount{#1,}{,}[\l] 
  \ifnum \l < 2 {\errmessage{#1 argument should be a comma-separated list of lenght >= 2}}
  \else{
    \foreach \n in {#1}{
      \path (\n);
      \pgfgetlastxy{\xnode}{\ynode};
      \pgfmathtruncatemacro{\xnround}{\xnode};
      \pgfmathsetmacro\tmpdiff{abs(\xnround - \xgateround)}
      \ifdim \tmpdiff pt > 1 pt \path[line] (\n) |- (#5);
      \else
        \path[line] (\n) -- (#5);
      \fi
    }
  }
  \fi
}

\newcommand{\gatedistancein}{3pt}
\newcommand{\gatedistanceinand}{2pt}

\usetikzlibrary{
  arrows,
  backgrounds,
  chains,
  calc,
  decorations.markings,decorations.pathreplacing,
  fadings,
  fit,
  patterns,
  petri,
  positioning,
  shadows,
  shapes,automata,shapes.callouts
}

\tikzset{
  src/.style={draw,circle,fill=white,
    minimum size=2mm,
    inner sep=0pt
  },
  sink/.style={draw,circle,double,fill=white,
    minimum size=1.5mm,
    inner sep=0pt
  },
  node/.style={draw,circle,fill=black,
    minimum size=2mm,
    inner sep=0pt
  },
  source/.style={draw,circle,fill=white,
    minimum size=3mm,
    inner sep=0pt
  },
  sink/.style={draw,circle,double,fill=white,
    minimum size=3mm,
    inner sep=0pt
  },
  block/.style = {rectangle, draw=gray, align=center, fill=orange!25, rounded corners=0.1cm,
    minimum size=5mm, inner sep=2pt},
  prenode/.style = {minimum size=9pt,inner sep=2pt, font=\Large},
  bblock/.style = {rectangle, draw=blue!50, opacity=.7, line width=.5pt, align=center, fill=white, rounded corners=0.1cm,
    minimum size=4mm, inner sep=1pt},
  prenode/.style = {minimum size=9pt,inner sep=2pt, font=\Large},
  agate/.style={draw, rectangle,
    minimum size=3mm,
    inner sep=0pt,
    fill=orange!25,
    label={[red]center:$\mid$}
  },
  ogate/.style = {
    diamond, draw, fill=orange!25,
    minimum size=4mm,
    inner sep=0pt,
    label={[red]center:$+$}
  },
  lgate/.style = {
    diamond, draw, fill=orange!25,
    minimum size=4mm,
    inner sep=0pt,
    label={[red]center:$\circlearrowleft$}
    },
  altogate/.style = {
    diamond, draw,
    minimum size=4mm,
    inner sep=0pt,
    postaction={path picture={%
        \draw
        ([yshift=\gatedistancein]path picture bounding box.south) -- ([yshift=-\gatedistancein]path picture bounding box.north)
        ([xshift=-\gatedistancein]path picture bounding box.east) -- ([xshift=\gatedistancein]path picture bounding box.west)
        ;}}},
  altgate/.style={draw, rectangle,
    minimum size=3mm,
    inner sep=0pt,
    postaction={path picture={%
        \draw
        ([yshift=\gatedistanceinand]path picture bounding box.south) --
        ([yshift=-\gatedistanceinand]path picture bounding box.north) ;}}},
  anygate/.style = {circle, draw, fill=white,
    minimum size=4mm,
    inner sep=0pt,
    postaction={path picture={%
        \draw[black]
        ([xshift=-\gatedistancein,yshift=\gatedistancein]path picture bounding box.south east) --
        ([xshift=\gatedistancein,yshift=-\gatedistancein]path picture bounding box.north west)
        ([xshift=-\gatedistancein,yshift=-\gatedistancein]path picture bounding box.north east) --
        ([xshift=\gatedistancein,yshift=\gatedistancein]path picture bounding box.south west)
        ;}}
  },
  smallglobal/.style={
        node distance=1cm and 0.8cm, semithick, scale=0.8, every node/.style={transform shape}
  },
  elli/.style = {draw,densely dotted,-},
  line/.style = {draw,->, rounded corners=0.07cm,>=latex},
  nline/.style = {draw,semithick, ->},
  pline/.style = {draw,->,>=latex},
  node distance=1cm and 0.7cm,
  baseline=(current  bounding  box.center),
  local/.style={rectangle, draw, fill=\fillcolor, drop shadow,
    text centered, rounded corners, minimum height=5em
  },
  bigar/.style={
    draw,very thick, ->
  },
  process/.style={rectangle, draw=gray, fill=\fillcolor, drop shadow,
    text centered, minimum height=5em,text=gray
  },
  choreo/.style={rectangle, draw, fill=\fillcolor, drop shadow,
    text centered, rounded corners, minimum height=5em
  },
  mycfsm/.style={
        font=\footnotesize,
        initial where=above,
        ->,>=stealth,auto, node distance=1cm and 1cm,
        scale=1, every node/.style={transform shape},
        every state/.style=inner sep=2pt,
        baseline=(current  bounding  box.center),
        initial text={}
  },
  machinecloud/.style={
    cloud, cloud puffs=10, cloud ignores aspect, minimum height=.1cm, minimum width=2cm, draw
  },
  fitting node/.style={
    inner sep=0pt,
    fill=none,
    draw=none,
    reset transform,
    fit={(\pgf@pathminx,\pgf@pathminy) (\pgf@pathmaxx,\pgf@pathmaxy)}
  },
  mypetri/.style={
    font=\footnotesize,
    baseline=(current  bounding  box.center)
  },
  silentrans/.style = {rectangle, draw=black, align=center, fill=black,
    minimum height=1pt,
    minimum width=15pt,
    inner sep=1.5pt
  },
  reset transform/.code={\pgftransformreset},
  tmtape/.style={draw,minimum size=1.2cm}
}

\newcommand{\gunlessop}{\mbox{\colorOp\tiny\tt unless}}

\newcommandx{\gtry}[5][1=\gname,2={\aG_1 \gchoop \cdots \gchoop \aG_n},3=\gin,4=\gout,5={j},usedefault=@]{
  \def\foo{\gtryop\ {#2} \ \gcatchop\ {#3} {\colorOp \Rightarrow} {#4} {\colorOp \bullet} {\gname[{#5}]}}
  \gnode[{#1}][{\ifempty{#1} {\foo } {(\foo)}}]
}

\newcommandx{\gtrycatch}[4][1=\gname,2={\aG},3=\gin,4={\aG'},usedefault=@]{
  \def\foo{\gtryop\ {#2} \ \gcatchop\ {#3} \gdoop\ {#4}}
  \gnode[{#1}][{\ifempty{#1} {\foo} {(\foo)}}]
}

\newcommandx{\agG}[2][1={\aG},2=\aguard]{{#1} \ifempty{#2}{}{\ \gunlessop\ {#2}}}

\newcommandx{\grcho}[5][1=\gname,2={\agG},3={\agG[\aG'][\aguard']},4={\cdots},5=A,usedefault=@]{
  \def\foo{{#2} {\ \ifempty{#4}{\gchoop}{\gchoop \ifempty{#4}{}{\ {#4}\  \gchoop}}\ } {#3}}
  \ifempty{#1}{\ifempty{#5}{\foo}{\gselop\ \cpt[{#1}][{\ptp[#5]}]\big\{ \foo \big\}}}{\gselop\ \cpt[{#1}][{\ptp[#5]}]\big\{ \foo \big\}}
}

\newcommandx{\ggprefix}[3][1=\ptp,2={\aR},3={\aR'},usedefault=@]{f_{#1}} 
\newcommand{\aconfigfn}{\chi}
\newcommand{\aconfig}{\ell}

\newcommand{\lstates}{\statemap}
\newcommandx{\sysconfig}[3][1=\lstates,2=\aconfigfn,3={},usedefault=@]{
  \conf{ {#1},{#2} \ifempty{#3}{}{, #3} }
}
\newcommand{\sysctxfn}[1][]{\gamma_{#1}}
\newcommandx{\sysctx}[2][1=\aQ,2={},usedefault=@]{({#1},\sysctxfn[{#2}])}

\newcommandx{\alog}[4][1=\msg,2=q,3=\gname,4=t,usedefault=@]{({#1},{#2},{#3},{#4})}

\newcommand{\aCM}{M}\newcommand{\aM}{\aCM}
\newcommand{\aQ}{Q}
\newcommandx{\aQzero}[1][1=,usedefault=@]{
  {\ifempty{#1}{q_0}{q_{0#1}}}
}
\newcommand{\badbranches}[1][]{\beta\ifempty{#1}{}{({#1})}}
\newcommand{\aTrs}{\tset}
\newcommandx{\guardedaction}[2][1=\al,2=\aguard,usedefault=@]{
  {#1} \ifempty{#2}{}{/} {#2}
}
\newcommandx{\atrM}[4][1=q,2=\al,3={\hat q,\hat \al, \aguard},4=q',usedefault=@]{
  {#1} \xrightarrow[{#3}]{\guardedaction[{#2}][]} {{#4}}
}
\newcommandx{\atrS}[5][
  1={\sysconfig[@][@][\badbranches]},
  2=\al,
  3=\aguard,
  4={\sysconfig[\lstates'][\aconfigfn'][\badbranches]},
  5=\sysctx,usedefault=@
]{
  {#1} \xRightarrow{\qquad} {{#4}}
}
\newcommandx{\arevtrS}[2][
  1={\sysconfig[@][@][\badbranches]},
  2={\sysconfig[\lstates'][\aconfigfn'][\badbranches']},
  usedefault=@
]{
  {#1} \rightsquigarrow {#2}
}
\newcommand{\aCS}{S}
\newcommand{\abuffer}{\vec b}

\newcommandx{\enables}[2][1=\aconfigfn,2=\aguard,usedefault=@]{{#1} \vdash {#2}}
\newcommandx{\gprojfn}[5][1=\aG,2=\ptp,3=\cinit,4=\cfinal,5={},usedefault=@]{
  \mathbf{proj}_{#2}({#1},{#3},{#4}\ifempty{#5}{}{,{#5}})
}

\newcommandx{\rbp}[3][1=\aG,2=\aconfigfn,3=\achan,usedefault=@]{\mathtt{RBP}_{{#1},{#2}}\ifempty{#3}{}{({#3})}}

\newcommand{\apseudoCFSM}{\mathtt{M}}
\newcommandx{\pseudoseq}[2][1=\apseudoCFSM,2=\apseudoCFSM',usedefault=@]{{#1}  ; {#2}}
\newcommandx{\pseudoCFSM}[4][1=\aQ,2=\aQzero,3=\cfinal,4=\aTrs,usedefault=@]{(#1 \ ; #2 \ ; #3 \ ; #4)}
\newcommandx{\markt}[3][1=\hat{\al},2=\hat{q},3=\aguard,usedefault=@]{\%\big({#1} , {#2}, {#3}\big)}

\newcommandx{\borderfn}[2][1=\aconfig,2=\aloop,usedefault=@]{
  \mathsf{border}_{{#2}}\ifempty{#1}{}{({#1})}
}

\tikzset{
  mycallout/.style={
	 fill=gray!30, opacity=.5, overlay, align=center,
	 cloud callout, cloud puffs=10, aspect=1.9, cloud ignores aspect, cloud puff arc=100
  }
}

\newcommandx{\ggvisually}[8][1=5pt,2=15pt,3=5pt,4=5pt,5=1.0cm,6=\scriptsize,7={},8={},usedefault=@]{
  \def\dist{\hspace{#5}}
  $\begin{array}{c@{\dist}c@{\dist}c@{\dist}c@{\dist}c@{\dist}c}
	  \begin{tikzpicture}[node distance=0.9cm and 0.4cm, every node/.style={scale=.7,transform shape}]
		 \node[source] (srcint) {};
		 \node[sink,below=of srcint] (sinkint) {};
		 \node[mycallout, above = .3cm of srcint, xshift=1cm, callout absolute pointer={(srcint.east)}] {source node};
		 \node[mycallout, below = .3cm of sinkint, xshift=1cm, callout absolute pointer={(sinkint.west)}] {sink node};
		 \path[line] (srcint) -- (sinkint);
	  \end{tikzpicture}
	  &
		 \begin{tikzpicture}[node distance=0.9cm and 0.4cm, every node/.style={scale=.7,transform shape}]
			\mkint{}{int}[]
			\mkgraph{int}{int};
		 \end{tikzpicture}
	  &
		 \begin{tikzpicture}[node distance=.9cm and 0.4cm, every node/.style={scale=.7,transform shape}]
			\node[bblock] at (0,0) (g) {$\aG$};
			\node[node, below=of g] (s1) {};
			\node[bblock, below=of s1] (gp) {$\aG'$};
			\path[line,dotted] (g) -- (s1);
			\path[line,dotted] (s1) -- (gp);
		 \end{tikzpicture}
	  &
		 \begin{tikzpicture}[node distance=.4cm and 0.4cm, every node/.style={scale=.7,transform shape}]
			\node[bblock] at (-.7,0) (g) {$\aG$};
			\node[bblock] at (.7,0)  (gp) {$\aG'$};
			\node[node, above=of g] (f) {};
			\node[node, below=of g] (j) {};
			\node[node, above=of gp] (fp) {};
			\node[node, below=of gp] (jp) {};
			\path[line,dotted] (f) -- (g);
			\path[line,dotted] (g) -- (j);
			\path[line,dotted] (fp) -- (gp);
			\path[line,dotted] (gp) -- (jp);
			\mkfork{f,fp}[fork][][#1];
			\mkjoin{j,jp}[join][][#2];
			\mkgraph{fork}{join};
			\node[mycallout, above = .3cm of fork, xshift=-1cm, callout absolute pointer={(fork.west)}] {fork gate};
			\node[mycallout, above = -.9cm of join, xshift=-1cm, callout absolute pointer={(join.west)}] {join gate};
		 \end{tikzpicture}
	  &
		 \begin{tikzpicture}[node distance=.4cm and 0.4cm, every node/.style={scale=.7,transform shape}]
			\node[bblock] at (-.7,0) (g) {$\aG$};
			\node[bblock] at (.7,0)  (gp) {$\aG'$};
			\node[node, above=of g] (f) {};
			\node[node, below=of g] (j) {};
			\node[node, above=of gp] (fp) {};
			\node[node, below=of gp] (jp) {};
			\path[line,dotted] (f) -- (g);
			\path[line,dotted] (g) -- (j);
			\path[line,dotted] (fp) -- (gp);
			\path[line,dotted] (gp) -- (jp);
			\mkbranch{f,fp}[fork][][#3];
			\mkmerge{j,jp}[join][][#4];
         \mkgraph{fork}{join};
         \node[mycallout, above = .3cm of fork, xshift=-1cm, callout absolute pointer={(fork.west)}] {branch gate};
         \node[mycallout, above = -.9cm of join, xshift=-1cm, callout absolute pointer={(join.west)}] {merge gate};
       \end{tikzpicture}
     \ifempty{#7}{}{
     &
		 \begin{tikzpicture}[node distance=0.4cm and 0.4cm, every node/.style={scale=.7,transform shape}]
        \node[bblock] (g) {$\aG$};
        \node[node, above=.5cm of g] (f) {};
        \node[node, below=.5cm of g] (j) {};
        \path[line,dotted] (f) -- (g);
        \path[line,dotted] (g) -- (j);
        \mkloop[.4][1]{f}{j};
        \mkgraph[.3cm]{entryf}{exitj};
        \node[mycallout, above = .2cm of entryf, xshift=1.3cm, callout absolute pointer={(entryf.east)}] {loop entry};
        \node[mycallout, above = -.7cm of exitj, xshift=1.3cm, callout absolute pointer={(exitj.west)}] {loop exit};
      \end{tikzpicture}
     }
     \ifempty{#8}{}{
	  \\
     \text{#6 empty}
     &
     \text{#6 interaction}
     &
     \text{#6 sequential}
     &
     \text{#6 parallel}
     &
     \text{#6 branch}
		 &\text{#6 iteration}
   }
   \end{array}$
}

\newcommandx{\newggvisually}[5][1=5pt,2=15pt,3=\scriptsize,4={},5={},usedefault=@]{
  \def\w{1cm}
  \begin{minipage}[c]{\w}
	 \ifempty{#5}{}{\text{#3 empty}\\[#1]}
	 \begin{tikzpicture}[node distance=0.9cm and 0.4cm, every node/.style={scale=.7,transform shape}]
		\node[source] (srcint) {};
		\node[sink,below=of srcint] (sinkint) {};
		\node[mycallout, above = .3cm of srcint, xshift=1cm, callout absolute pointer={(srcint.east)}] {source node};
		\node[mycallout, below = .3cm of sinkint, xshift=1cm, callout absolute pointer={(sinkint.west)}] {sink node};
		\path[line] (srcint) -- (sinkint);
	 \end{tikzpicture}
  \end{minipage}
  \hfill
	  \begin{minipage}[c]{\w}
     \ifempty{#5}{}{\text{#3 interaction}\\[#1]}
		 \begin{tikzpicture}[node distance=0.9cm and 0.4cm, every node/.style={scale=.7,transform shape}]
			\mkint{}{int}[]
			\mkgraph{int}{int};
		 \end{tikzpicture}
	  \end{minipage}
	 \hfill
	  \begin{minipage}[c]{.1cm}
     \ifempty{#5}{}{\text{#3 sequential}\\[#1]}
		 \begin{tikzpicture}[node distance=.9cm and 0.4cm, every node/.style={scale=.7,transform shape}]
			\node[bblock] at (0,0) (g) {$\aG$};
			\node[node, below=of g] (s1) {};
			\node[bblock, below=of s1] (gp) {$\aG'$};
			\path[line,dotted] (g) -- (s1);
			\path[line,dotted] (s1) -- (gp);
		 \end{tikzpicture}
	  \end{minipage}
	  \hfill
	  \begin{minipage}[c]{\w}
     \ifempty{#5}{}{\text{#3 parallel}\\[#1]}
		 \begin{tikzpicture}[node distance=.4cm and 0.4cm, every node/.style={scale=.7,transform shape}]
			\node[bblock] at (-.7,0) (g) {$\aG$};
			\node[bblock] at (.7,0)  (gp) {$\aG'$};
			\node[node, above=of g] (f) {};
			\node[node, below=of g] (j) {};
			\node[node, above=of gp] (fp) {};
			\node[node, below=of gp] (jp) {};
			\path[line,dotted] (f) -- (g);
			\path[line,dotted] (g) -- (j);
			\path[line,dotted] (fp) -- (gp);
			\path[line,dotted] (gp) -- (jp);
			\mkfork{f,fp}[fork][][#1];
			\mkjoin{j,jp}[join][][#2];
			\mkgraph{fork}{join};
			\node[mycallout, above = .2cm of fork, xshift=-1cm, callout absolute pointer={(fork.west)}] {#3 fork gate};
			\node[mycallout, below = .2cm of join, xshift=-1cm, callout absolute pointer={(join.west)}] {#3 join gate};
		 \end{tikzpicture}
	  \end{minipage}
	  \hfill
	  \begin{minipage}[c]{\w}
     \ifempty{#5}{}{\text{#3 branch}\\[#1]}
		 \begin{tikzpicture}[node distance=.4cm and 0.4cm, every node/.style={scale=.7,transform shape}]
			\node[bblock] at (-.7,0) (g) {$\aG$};
			\node[bblock] at (.7,0)  (gp) {$\aG'$};
			\node[node, above=of g] (f) {};
			\node[node, below=of g] (j) {};
			\node[node, above=of gp] (fp) {};
			\node[node, below=of gp] (jp) {};
			\path[line,dotted] (f) -- (g);
			\path[line,dotted] (g) -- (j);
			\path[line,dotted] (fp) -- (gp);
			\path[line,dotted] (gp) -- (jp);
			\mkbranch{f,fp}[fork][][#1];
			\mkmerge{j,jp}[join][][#2];
         \mkgraph{fork}{join};
         \node[mycallout, above = .1cm of fork, xshift=-1cm, callout absolute pointer={(fork.west)}] {branch gate};
         \node[mycallout, below = .1cm of join, xshift=-1cm, callout absolute pointer={(join.west)}] {merge gate};
       \end{tikzpicture}
	  \end{minipage}
     \ifempty{#4}{}{
     \hfill
	  \begin{minipage}[c]{\w}
     \ifempty{#5}{}{\text{#3 iteration}\\[#1]}
		 \begin{tikzpicture}[node distance=0.4cm and 0.4cm, every node/.style={scale=.7,transform shape}]
        \node[bblock] (g) {$\aG$};
        \node[node, above=.5cm of g] (f) {};
        \node[node, below=.5cm of g] (j) {};
        \path[line,dotted] (f) -- (g);
        \path[line,dotted] (g) -- (j);
        \mkloop[.4][1]{f}{j};
        \mkgraph[.3cm]{entryf}{exitj};
        \node[mycallout, above = .2cm of entryf, xshift=1.3cm, callout absolute pointer={(entryf.east)}] {loop entry};
        \node[mycallout, above = -.7cm of exitj, xshift=1.3cm, callout absolute pointer={(exitj.east)}] {loop exit};
      \end{tikzpicture}
	  \end{minipage}
     }
}

  \newcommandx{\wwwcquote}[1][1=quo:w3c,usedefault=@]{
	 \ifempty{#1}{}{\begin{quote}\label{#1}}
		\lq\lq Using the Web Services Choreography specification, a
		\textcolor{orange}{contract} containing a global definition of the
		common \textcolor{orange}{ordering} conditions and constraints
		under which \textcolor{orange}{messages} are exchanged, is
		produced that describes, from a \textcolor{orange}{global
		  viewpoint} [...]  observable behaviour of all the parties
		involved.
		\textcolor{OliveGreen}{Each party} can then use the global definition to
		\textcolor{OliveGreen}{build and test solutions that conform to it}.
		The global specification is in turn \textcolor{OliveGreen}{realised by combination of} the
		resulting \textcolor{OliveGreen}{local systems} [...]\rq\rq
		\ifempty{#1}{}{\end{quote}}
  }


%% file: macros.tex
\newcommand{\subject}{\esbj}

\newcommand{\varset}{\mathtt{X}}
\newcommandx{\constset}[1][1=C,usedefault=@]{\mathsf{#1}}
\newcommandx{\termalg}[2][1=\Sigma,2=\varset,usedefault=@]{\mathsf{Term}_{#1,#2}}
\newcommandx{\formulas}[2][1=\Sigma,2=\varset,usedefault=@]{\mathsf{Form}_{#1,#2}}
\newcommandx{\modelsclass}[2][1=\Sigma,2=\varset,usedefault=@]{\mathsf{Mod}_{#1,#2}}

\newcommand{\QL}{\mathcal{QL}}

\newcommand{\DLTL}{\mathit{DLTL}}
\newcommand{\gqosprop}{\Phi}

\newcommand{\genericAttr}{\mathsf{a}}
\newcommand{\qosspecs}{\mathcal{C}}

\newcommand{\gpop}[1]{\aG_\mathrm{#1}}

\mkfun{\aggfn}{agg}{}
\mkfun{\qos}{qos}{}
\mkfun{\splitfn}{split}{}

\def \({\left (}
\def \){\right )}
\def \[{\left [}
\def \]{\right ]}
\def \[[{\left \llbracket}
\def \]]{\right \rrbracket}
\def \<{\left\langle}
\def \>{\right\rangle}

\newcommandx{\psiinst}[3][1=\psi,2=a,3=q,usedefault=@]{{#1}_{\ptp[#2]}^{{#3}}}
\let\vec\mathfrak


%% file: intro.tex
%
Over the past two decades, software has steadily changed from monolithic 
applications to distributed cooperating components.
Choreographic approaches are gaining momentum in industry
(e.g.~\cite{bpmn,bon18,fmmt20,DBLP:journals/software/AutiliIT15}) which,
increasingly, conceives applications as components
interacting over existing communication infrastructures.
Among other models, choreographies stand out for a neat separation of
concerns: choreographic models abstract away local computations from communications among participants.
In fact, since their introduction~\cite{w3c:wsdl20}, choreographies
advocate for a separation between a \emph{global view} and a \emph{local view} of communication.
The former is a high-level description of (distributed) interactions.
The latter view is a description of each component in isolation.
This is the distinctive feature of choreographies that we exploit here
to reason about quantitative properties of applications.
The basic idea is to specify the values of quality attributes of
local states of components and then \emph{aggregate} those attributes 
along runs involving communications.
A simple example can illustrate this.
Suppose that a component \p\ sends another component \q\ a message
$\msg$ and we want to consider two quality attributes: monetary cost
($c$) and memory consumption ($\mathit{mem}$).
This behaviour can be abstracted away with the finite-state machines below
\begin{equation}\label{eq:idea}
  \begin{array}{l@{:\qquad}l}
  \text{behaviour of \p} &
  \begin{tikzpicture}[node distance = 1cm and 4cm, every node/.style={font=\scriptsize}, every label/.style={color=blue,fill=yellow!10}]
	 \node[label={above:$\{ c \leq 5, \mathit{mem} = 0 \}$}] (aq0) {$\aQzero$};
    \node[right = of aq0, label={above:$\{ 5 \leq c \leq 10,\ \mathit{mem} < 3 \}$}] (aq1){$q_1$};
	 \path[->,draw] (aq0) edge node[above] {$\aout$} (aq1);
  \end{tikzpicture}
  \\[2em]
  \text{behaviour of \q} &
  \begin{tikzpicture}[node distance = 1cm and 4cm, every node/.style={font=\scriptsize}, every label/.style={color=blue,fill=yellow!10}]
    \node[label={above: $\{c=0,\mathit{mem}=0\}$}] (bq0){$\aQzero'$};
    \node[right = of bq0, label={above:$\{10 \leq \mathit{mem} \leq 50,\ c = 0.01 \cdot \mathit{mem} \}$}] (bq1){$q_1'$};
	 \path[->,draw] (bq0) edge node[above] {$\ain$} (bq1);
  \end{tikzpicture}
  \end{array}
\end{equation}
where $\aout$ and $\ain$ respectively denote the output and input
communication actions, and each state is decorated with a specification
predicating over the quality attributes in the local states of \p\ and \q.
For instance, both \p\ and \q\ allocate no memory in their initial
states, computation in \p\ may cost up to five monetary units before 
executing the output, and \q\ has no cost since it's just waiting to 
execute the input ($c = 0$).
Likewise, after the communication actions, the local computations
of \p\ and \q\ are specified by the formulae associated to states $q_1$
and $q_1'$.

The interaction between \p\ and \q\ depends on the communication
infrastructure; e.g., asynchronous message-passing yields a run like
\begin{align*}
  \pi: \quad
  \begin{tikzpicture}[node distance = 1cm and 4cm, every node/.style={font=\scriptsize}]
	 \node(s0){$\aConf_0$};
	 \node[right = of s0] (s1){$\aConf_1$};
	 \node[right = of s1] (s2){$\aConf_2$};
	 \path[->,draw] (s0) edge node[above] {$\aout$} (s1);
	 \path[->,draw] (s1) edge node[above] {$\ain$} (s2);
  \end{tikzpicture}
\end{align*}
where first the message is sent by \p\ and then it is eventually
received by \q.

We are interested in analyzing quality properties that admit a measurement,
thus assuming that the QoS attributes are \emph{quantitative}.
These properties encompass both quantitative attributes at the application 
level as well as resource consumption metrics.
For instance, we could be interested in analyzing the monetary cost or 
the number of messages retrieved in a messaging system; but we could
also be interested in analyzing its memory usage or CPU time.
It's important to emphasize that our framework is designed to be 
agnostic and adaptable, allowing for the consideration of any 
quantifiable attribute, regardless of its specific nature.
Furthermore, our framework is specifically designed to enable analysis 
of how quantitative properties of local computations 
influence the system-wide properties.
Hence, we envisage the quality constraints as \emph{contracts} that
the local computations of components should honour.
For instance, the specifications in~\eqref{eq:idea} tell the cost
of local computations in \p\ and \q, they do not predicate on the 
QoS of the communication infrastructure.

Once these quality constraints on local computations are fixed,
natural questions to ask are e.g., \quo{is the memory consumption of
  \q\ along run $\pi$ within a given range?} or \quo{is the overall
ry cost below a given threshold?}.
Answers to such questions require checking that the \emph{aggregation}
of the values of the quality attributes along the run $\pi$
entails the properties.
Interestingly, how to aggregate those values depends on the quality
attributes.
For instance, the aggregation of memory consumption can be computed by
taking the maximum, while the aggregation of monetary cost can be
computed as the sum.
%
We work under the hypothesis that developers have no control
over communication infrastructure.
More precisely, QoS aspects related to how communications are realised
are not under the control of applications' designers.
Instead, designers have control over local computations, thus suggesting
that QoS constraints are naturally associated to states of components.
Indeed, we rely on behavioural types 
(such as~\cite{gpsy16,adsgpt19,bmt19,DBLP:conf/popl/BasuBO12,hlvccdmprt16}) which 
abstract away low level details.

\paragraph{Contributions}
We propose a framework for the design and analysis of QoS-aware
distributed systems, enabled by the following technical contributions:
\begin{description}
\item[Models for QoS attributes.]
  \cref{sec:QoS-choreo} presents a straightforward extension of
  communicating finite-state machines (CFSMs~\cite{brand:jacm-30_2};
  reviewed in \cref{sec:pre}) to express QoS aspects of components.
  Basically, we assign to each state of CFSMs a QoS specification
  as in~\eqref{eq:idea}.
  We adopt real-closed fields (RCFs, cf.~\cref{sec:pre}) to abstractly
  represent QoS values; besides being a complete and decidable
  abstract formalisation of the first-order theory of the real numbers,
  RCFs are instrumental for a smooth definition of our framework.
\item[A dynamic temporal logic for QoS.]  \cref{sec:logic} introduces a logic, dubbed
  $\QL$, to express QoS properties.
  Taking inspiration from \emph{Propositional Dynamic Linear Temporal
	 Logic} (DLTL)~\cite{henriksen:apal-96_1_3}, $\QL$ indexes temporal
  modalities with \emph{global choreographies}~\cite{tuosto:jlamp-95}
  (g-choreographies, \cref{sec:pre}), a model of global views of
  choregraphies, in order to predicate over QoS properties of the whole
  system.
  This is a distinct characteristic of $\QL$ that we comment
  in~\cref{sec:rw}.
\item[A semi-decision procedure for $\QL$.] \cref{sec:analysis}
  proves $\QL$ to be semi-decidable by providing a $k$-bounded
  semi-decision procedure and relying on the decidability of the
  theory of real-closed fields \cite{tarski:RM-109} to check QoS 
  constraints in atomic formulae.
  A distinctive aspect of the procedure is that it can be used 
  as a bounded model-checking procedure 
  of $\QL$ formulae.
\end{description}
\cref{sec:conclu} draws some conclusions and points out some further
lines of research.
%


%% file: rw.tex
The relevance of the problem addressed here has been already
highlighted by other researchers~\cite{ich12,kgi13}.
There is a vast literature on QoS, spanning a wide range of
contexts and methods~\cite{Aleti2013SoftwareAO,hayyolalam:jnca-110}.
This paper can be positioned in the category of general
application-level QoS.
The combination of RCFs and our behavioural types aims to capture
essential aspects for applications' quantitative analysis while
striving for generality.
In this vein, a proof-of-concept methodology based on behavioural
types has been proposed in~\cite{gdgln16} for client-server systems.
To the best of our knowledge, our is the first work blending
behavioural types with QoS and offering a decision procedure for
multiparty protocols.

In order to abstractly capture QoS (instead of focusing on specific
attributes) we adopt RCFs.
Other abstract models of QoS such as quantales~\cite{Rosenthal90} or
c-semirings~\cite{buscemi:esop07,lm05,dfmpt05} have been proposed.
We opted for RCFs due to their inherent decidability, which is 
crucial for ensuring the decidability of our logic. Moreover, 
RCFs offer practical advantages as they can be readily employed 
in modern SMT (satisfiability modulo theories) solvers \cite[Chapter 33]{biere21}.

Theory presentations over QoS attributes are used
in~\cite{martinezsune:coordination19} to enable the automatic analysis
of QoS properties with a specification language that only considers convex
polytopes; this restriction is not present in our language.
Also, the approach
in~\cite{martinezsune:coordination19} can be thought as
\quo{monolithic}, in the sense that specifications are given
considering the system as a black box.
We instead assign QoS contracts to states of components and then
aggregate them in order to analyze properties along executions of the
behavior emerging from interactions.

The use of choreographic methods for non-functional analysis yields
other advantages.
For instance, QoS contracts of components are derived from global
specifications~\cite{ich12}.
These contracts can then be used for run-time prediction, adaptive
composition, or compliance checking, similarly to what is done
in~\cite{kgi13}.
This top-down approach can be transferred to behavioural types as well
similarly to what has been done in~\cite{bmt19,bhty10} for qualitative
properties.
The framework proposed in~\cite{vissani:places15} uses CFSMs as a 
dynamic binding mechanism of services but only considers the 
communicational aspects of the software component. Such a framework
could be extended to include QoS attributes as well by leveraging
the results presented in this paper.

Our $\QL$ logic takes inspiration from \emph{dynamic linear temporal
  logic} ($\DLTL$)~\cite{henriksen:apal-96_1_3} which blends
trace semantics (akin \emph{linear temporal logic}~\cite{pnueli:tcs-13_1}) 
 and regular expressions over a set of atomic actions (akin programs in \emph{propositional
  dynamic logic}~\cite{pratt:ieee-sfcs76}). 
Intuitively a key difference is that, unlike $\DLTL$, $\QL$
does not predicate about the behaviour of sequential programs; rather
$\QL$ describes properties of asynchronous message-passing systems.
This requires a modification of the syntax of $\DLTL$; in
fact, the syntax of $\QL$ is essentially the same of $\DLTL$
barred for the indexes of modalities, which become choreographies of
interactions.
This straightforward modification has deep impact on the semantics
which requires a complete redefinition (see \cref{sec:logic} for
further details).
Another key difference is that, while $\DLTL$ is propositional, 
$\QL$'s atomic formulae are first order formulae on QoS attributes.
As a consequence, not only $\QL$ can express usual temporal properties, 
such as safety and liveness ones, but temporal properties constraining
the value of QoS attributes.
These points of comparison with $\DLTL$ apply in the 
same way to a similar logic called \emph{linear dynamic logic} 
($\mathsf{LDL}$), introduced first in \cite{vardi:eptcs-54} and later 
formalized for finite traces in \cite{degiacomo:ijcai13}.


%% file: preliminaries.tex

This section
surveys 
background material underpinning 
our work.
We first describe the protocol used as a running example, then we
review our choreographic model 
and we briefly
recall \emph{real-closed fields}. 

\noindent\textbf{A running example.\ }
\input{pop}

\noindent\textbf{A choreographic model.\ }
\input{chor}

\noindent\textbf{Real-closed fields.\ }
Real numbers are natural candidates to express quantitative attributes
of a software artifact.
We adopt \emph{real-closed fields} (RCFs), which is the formalisation of the
first-order theory of the real numbers, as a foundation for QoS
values.
Let $\Sigma_{\mathrm{field}}$ denote the first-order signature $\conf{\{0, 1\},\{+, \cdot\},\{<\}}$. An ordered field is a first-order theory presentation $\conf{\Sigma_{\mathrm{field}}, \Gamma_{\mathrm{field}}}$, where $\Gamma_{\mathrm{field}}$ consists of the \emph{field} axioms as well as the axioms defining $<$ as a strict total order relation.
\mbox{Real-closed} fields are ordered fields whose non-empty subsets all have
a supremum.
Tarski's axiomatization of real-closed fields, denoted here as $\conf{\Sigma_{\mathrm{RCF}}, \Gamma_{\mathrm{RCF}}}$, was introduced in \cite{tarski:RM-109}. Tarski further demonstrated the existence of a decision procedure for this first-order theory of real numbers in \cite[Thm.~37]{tarski:RM-109}.
%
Thus, the main reason for selecting RCFs as the foundation for QoS lies in the fact
that first-order theories extending them using elementary operations are decidable,
providing effective means for analysis. 



%% file: pop.tex
Through the paper we will use a (simplified variant) of the POP
protocol~\cite{rfc937}.
This protocol allows mail clients to access a remote mailbox and
retrieve e-mails.
In the POP protocol a client starts the communication by sending a
message of type $\msg[helo]$ to a POP server (note that protocol
specifications are oblivious of messages' payload).\footnote{Our
  framework can handle multiparty protocols; however, our examples are
  two-party for simplicity. Also, we stick to the types of messages as
  carefully described in the protocol specifications~\cite{rfc937}.  }
The server replies with the number of unread messages in the mailbox
using a message of type $\msg[int]$.
At this point the client can either halt the protocol or
read one of the e-mails.
These options are selected by sending a message of type $\msg[quit]$
or of type $\msg[read]$ respectively.
In the former case, the server acknowledges with a message of type
$\msg[bye]$ and the protocol ends.
In the latter case, the server sends the client the number of bytes of
the current unread message in a message of type $\msg[size]$.
Next, the client has again a choice between quitting the protocol (as
before) or receiving the email (selected in the $\msg[read]$ message)
by sending the server a message of type $\msg[retr]$.
In the latter case the server sends the email with a message of type
$\msg[msg]$, the client answers with a message of type $\msg[ack]$
and the reading process starts again.

%


%% file: chor.tex
%
We use \emph{global choreographies}~\cite{tuosto:jlamp-95} to specify
the global view of communicating systems whose local view are rendered
as \emph{communicating finite state machines}~\cite{brand:jacm-30_2}.

Hereafter, we fix a set $\PSet$ of \emph{participants} and a set
$\msgset$ of (types of) \emph{messages} such that
$\PSet \cap \msgset = \emptyset$.
We start by surveying the definition of g-choreographies.

\begin{definition}[Global choreographies~\cite{tuosto:jlamp-95}]\label{def:g-choreography}
  A \emph{global choreography over $\PSet$ and $\msgset$}
  (\emph{g-choreography} for short) is a term $\aG$ that can be
  derived in
  \begin{align*}
	 \aG \bnfdef \gempty \bnfmid \gint[] \bnfmid \gseq[] \bnfmid \gpar[]
	 \bnfmid \gcho[] \bnfmid \grec
  \end{align*}
  where $\p, \q \in \PSet$, $\p \neq \q$ and $\msg \in \msgset$.
\end{definition}
Intuitively, a g-choreography specifies the communication protocol of
participants.
The basic g-choreography is the empty one $\gempty$, which specifies
that no communications should happen.
An \emph{interaction} $\gint[]$ specifies that participants $\p$ and
$\q$ (are expected to) exchange a message of type $\msg$;
  it is worth remarking that we assume asynchronous communication
  where the sender $\p$ does not wait for $\q$ to consume $\msg$
  to continue its execution.
Moreover, g-choreographies can be composed sequentially ($\gseq[]$),
in parallel ($\gpar[]$), and in non-deterministic choices
($\gcho[]$); we assume that $\gempty$ is the neutral element
  of $\gseq[][\_][\_]$, $\gpar[][\_][\_]$, and $\gcho[][\_][\_]$.
  Note that, due to asynchrony in the communication, in a sequential
  composition $\gseq[]$, outputs in $\aG'$ can occur before $\aG$ is
  fully executed; for instance, the distributed execution of
  $\gseq[][{\gint[]}][{\gint[][c][m'][b]}]$ allows the output from
  $\ptp[c]$ to happen before the one from $\p$.
Finally, a g-choreography may be iterated ($\grec[]$).

\begin{example}[A g-choreography for POP]\label{ex:gpop}
  Our running example can be expressed as the g-choreography
  $\gpop{pop} =
  \gcho[][{\gseq[][{\gint[][c][helo][s]}][{\gpop{start}}]}][{\gpop{quit}}]$
  where
  \begin{align*}
	 \begin{array}{l@{\quad}l}
		\gpop{start} = 
    \gseq[][
      {\gint[][s][int][c]}
      ][
        {\gseq[][
          {\grec[][({\gcho[][\gpop{read}][{\gseq[][\gpop{read}][\gpop{retr}]}]})]}
          ][
            {\gpop{quit}}
            ]
        }
        ]
		& \gpop{read} = \gseq[][{\gint[][c][read][s]}][{\gint[][s][size][c]}]
		\\
		\gpop{retr} = \gseq[][{\gint[][c][retr][s]}][{\gint[][s][msg][c];\gint[][c][ack][s]}]
		& \gpop{quit} = \gseq[][{\gint[][c][quit][s]}][{\gint[][s][bye][c]}]
  \end{array}
\end{align*}
($\gseq[][\_][\_]$ takes precedence over $\gcho[][\_][\_]$).
%
%
\finex
\end{example}

The participants of a communicating system interact through
\emph{channels} borrowed from the set
$\chset = \set{(\p, \q) \in \PSet \times \PSet \sst \p \neq \q}$.
A channel $(\p,\q) \in \chset$ (written $\achan$ for short) allows
$\p$ to asynchronously send messages to $\q$ through an unbounded FIFO
buffer associated to $\achan$.
The set of \emph{communication actions} is
$\lset = \lset^!  \cup \lset^?$ where
$\lset^! = \{ \aout \sst \achan \in \chset \text{ and } \msg \in
\msgset\}$ and
$\lset^? = \{ \ain \sst \achan \in \chset \text{ and } \msg \in
\msgset\}$ are respectively the set of \emph{output} and \emph{input}
actions.
%
The \emph{language} $\rlang[\aG]$ of a g-choreography $\aG$ 
is essentially the set of all possible sequences in $\lset$
compatible with the causal relation induced by $\aG$.
Since $\rlang[\aG]$ is prefix-closed, we write $\hat{\rlang}[\aG]$
for the set of sequences in $\rlang[\aG]$ that are not proper prefixes
of any other sequence in $\rlang[\aG]$.
The technical definition of $\rlang[\aG]$, immaterial here,
can be found in~\cite{tuosto:jlamp-95}. 
We will adapt CFSM~\cite{brand:jacm-30_2} to model the QoS-aware \emph{local
  view} of a system.

\begin{definition}[Communicating systems~\cite{brand:jacm-30_2}]\label{def:cfsm}
  A \emph{communicating finite-state machine} (\emph{CFSM})
  is a finite transition system $\aCM = (\aQ,\aQzero,\aTrs)$ where
  \begin{itemize}
  \item $\aQ$ is a finite set of {\em states} with $\aQzero \in \aQ$
	 the \emph{initial} state, and
  \item $\aTrs\ \subseteq \ \aQ \times \lset \times \aQ$; we write
	 $q \trans{\al} {q'}$ for $(q,\al,q') \in \aTrs$.
  \end{itemize}
  For $\aout \in \lset$ (resp.  $\ain \in \lset$), let
  $\subject[\aout] = \p$ (resp.  $\subject[\ain] = \q$).
  Given $\p \in \PSet$, $\aCM$ is \emph{$\p$-local} if
  $\subject[\al] = \p$ for each $q \trans{\al} {q'}$.
  A \emph{(communicating) system} is a map
  $\aCS = (\aCM_{\ptp})_{\ptp \in \ptpset}$ assigning a $\p$-local
  CFSM $\aCM_{\ptp}$ to each $\ptp \in \ptpset$.
\end{definition}

\begin{example}[Communicating system for POP]\label{ex:pop}
    The following CFSM exhibits a behaviour of a POP client
    compatible with the protocol in
    \cref{ex:gpop} because its executions yield a subset of 
    the client's execution specified there.
{\small  \begin{align*}
   \begin{tikzpicture}[node distance=1cm and 1.3cm, every node/.style={scale=1,transform shape}]
	 \node (d) at (0,0) {};
	 \node[cnode,right= .3cm of d] (0) {};
	 \foreach \s/\t in {0/1,1/2,2/3,3/4,4/5,5/6} {
		\node (\t) [cnode, right = of \s] {};
	 }
	 \foreach \s/\l/\t in {0/\aout[c][s][helo]/1,1/\ain[s][c][int]/2,2/\aout[c][s][read]/3,3/\ain[s][c][size]/4,4/\aout[c][s][retr]/5,5/\ain[s][c][msg]/6} {
		\path[line,sloped] (\s) edge node[below]{$\l$} (\t);
	 }
	 \node[cnode,above = of 0] (7) {};
	 \node[cnode,right = of 7] (8) {};
	 \path[line] (d) -- (0);
	 \path[line] (6) edge[bend left=-35] node[below]{$\aout[c][s][ack]$} (2);
	 \path[line] (0) edge node[left]{$\aout[c][s][quit]$} (7);
	 \path[line] (2) edge node[left]{$\aout[c][s][quit]$} (7);
	 \path[line] (4) edge[sloped] node[above]{$\aout[c][s][quit]$} (7);
	 \path[line] (7) edge node[above]{$\ain[s][c][bye]$} (8);
  \end{tikzpicture}
  \end{align*}
}  %
  For a POP server, being a two-party protocol, we can use a dual CFSM obtained by replacing
  send actions with corresponding receive actions and
  vice versa.
  \finex
\end{example}
%

The asynchronous communication between participants is formalised by a labelled transition system (LTS) tracking the
(local) state of each CFSM and the content of each buffer (i.e. communication channel) in the
system.
A \emph{configuration} of a communicating system $\aCS$ is a pair
$\aConf = \csconf q \abuffer$ where $\vec q$ and $\vec \abuffer$
respectively map participants to states and channels to sequences of
messages; state $\vec q(\p)$ keeps track of the state of the machine
$\aCM_{\p}$ and buffer $\abuffer(\achan)$ yields the messages
sent from $\p$ to $\p[b]$ and not yet consumed.
The \emph{initial} configuration $\aConf_0$ is the one where, for all
$\p \in \ptpset$, $\vec q(\p)$ is the initial state of the
corresponding CFSM and $\abuffer(\achan)$ is the empty sequence for
all $\achan \in \chanset$.

A configuration $\aConf'= \csconf {q'} {\abuffer'} $ is {\em
  reachable} from another configuration
$\aConf = \csconf {q} {\abuffer}$ by \emph{firing a transition $\al$},
written $\aConf \TRANSS{\al} \aConf'$, if there is a message
$\msg \in \msgset$ such that either (1) or (2) below holds:
\begin{center}
  \begin{tabular}{l@{\hspace{.45cm}}r}
    \begin{minipage}{.47\linewidth}\small
      1.  $\al = \aout[@][@][@]$ with
		$\vec q(\p) \trans[\p]{\al} {q'}$ and
        \begin{itemize}
        \item[a.] $\vec q' = \vec q[\p \mapsto q']$
        \item[b.] $\abuffer' = \abuffer[\achan \mapsto \abuffer(\achan).\msg]$
        \end{itemize}
    \end{minipage}
    &
    \begin{minipage}{.47\linewidth}\small
      2.
        $\al = \ain[@][@][@]$ with $\vec q(\q) \trans[\q]{\al} {q'}$
        and 
        \begin{itemize}
        \item[a.] $\vec q' = \vec q[\q \mapsto q']$ and
        \item[b.] $\abuffer = \abuffer'[\achan \mapsto \msg.\abuffer'(\achan)]$.
        \end{itemize}
    \end{minipage}
  \end{tabular}
\end{center}
Condition (1) puts $\msg$ on channel $\ptp\ptp[B]$, while (2) gets
$\msg$ from channel $\ptp\ptp[B]$.
In both cases, any machine or buffer not involved in the transition
is left unchanged in the new configuration $\aConf'$.
\begin{example}[Semantics of CFSMs]
  For the run $\pi$ of the communicating system in \eqref{eq:idea}
  (cf. \cref{sec:intro}) we have, for $i \in \{0,1,2\}$,
  $\aConf_i = \csconf{q_i}{\abuffer_i}$, where
  $\vec{q_0} = \{\p \mapsto \aQzero, \q \mapsto \aQzero'\}$,
  $\abuffer_0 = \{\achan \mapsto \epsilon,
  \achan[b][a] \mapsto \epsilon\}$,
  $\vec{q_1} = \{\p \mapsto q_1, \q \mapsto \aQzero'\}$,
  $\abuffer_1 = \{\achan \mapsto \msg, \achan[b][a] \mapsto
  \epsilon\}$, and $\vec{q_2} = \{\p \mapsto q_1, \q \mapsto q_1'\}$,
  $\abuffer_2 = \abuffer_0$.
  \finex
\end{example}
Let $\aCS$ be a communicating system.
A sequence $\pi = (\aConf_i,\al_i,\aConf_{i+1})_{i \in I}$ where
$I$ is an initial segment of natural numbers (i.e., $i-1 \in I$ for
each $0 < i \in I$) is a run of $\aCS$ if
$\aConf_i \TRANSS{\al_i} \aConf_{i+1}$ is a transition of $\aCS$ for
all $ i \in I$.
The set of runs of $\aCS$ is denoted as $\Delta^\infty_{\aCS}$
and the set of runs of length $k$ is denoted as $\Delta^k_{\aCS}$.
Note that $\Delta^\infty_{\aCS}$ may contain runs of infinite length, 
the set of finite runs of $\aCS$ is the union of all $\Delta^k_{\aCS}$ 
and will be denoted as $\Delta_{\aCS}$.
Given a run $\pi$,
 we define
$\rlang[\pi]$ to be the sequence of labels $(\al_i)_{i \in I}$.
The \emph{language} of $\aCS$ is the set
$\rlang[\aCS] = \{\rlang[\pi] \sst \pi \in
\Delta^\infty_{\aCS}\}$.
Finally, $\mathit{prf}: \Delta^\infty_{\aCS} \to 2^{\Delta_{\aCS}}$
maps each run $\pi \in \Delta^\infty_{\aCS}$ to its set of finite
prefixes.
As usual, for all $\pi \in \Delta^\infty_{\aCS}$,
the
empty prefix
$\epsilon$ belongs to $\mathit{prf} (\pi)$.
For convenience, we will occasionally write
$ \aConf_0 \TRANSS{\al_0} \aConf_1\ \ldots\
\aConf_n \TRANSS{\al_n} \aConf_{n+1}$ for finite sequences.



%% file: qoschoreo.tex

In this section we extend CFSMs with QoS specifications in order to
express QoS contracts of components in message-passing systems.
Basically, each state of CFSMs is assigned a QoS contract specifying
the usage of computational resources.
We formalise QoS contracts as \emph{QoS specifications}
which are theory presentations over the
RCFs, noted as $\conf{\Sigma, \Gamma}$, paired up with \emph{aggregation operators}, noted as $\oplus^\genericAttr$,
to define how each
QoS attribute accumulates along a communicating system. 
These aggregation operators will be essential to formally define the notion of aggregation along a run, as shown later in \cref{ex:aggregation-function}.
\begin{definition}\label{def:qos-spec}
  A \emph{QoS specification} $\conf{\Sigma, \Gamma}$ is a
  (first-order) theory presentation extending
  $\conf{\Sigma_\mathrm{RCF}, \Gamma_\mathrm{RCF}}$ as follows:
  \begin{enumerate}
  \item
	 $\Sigma = \conf{\{0, 1\} \cup \constset[Q], \{+, \cdot\} \cup
		\{\oplus^\genericAttr\}_{\genericAttr \in \constset[Q]}, \{<\}}$, where $\constset[Q]$
	 is a finite set of constant symbols (other that
	 $0$ and $1$) representing the \emph{quantitative attributes} (from
	 now on referred to as \emph{QoS attributes}) and, for each
	 $\genericAttr \in \constset[Q]$, $\oplus^\genericAttr$ is an associative
	 algebraic binary operator 
   and
  \item $\Gamma = \Gamma_\mathrm{RCF} \cup \text{$\Gamma^\prime$}$,
	 being $\Gamma^\prime$ a finite set of
	 first-order formulae formalising specific constraints over 
   the QoS attributes in $\constset[Q]$.
  \end{enumerate}
  The class of QoS specifications will be denoted as $\qosspecs(\constset[Q])$.
\end{definition}
In order to preserve decidability of QoS properties, we only consider
QoS specifications involving additional constant symbols representing
the QoS attributes of components.
Aggregation operators are required to be algebraic 
because the extension of the theory must be kept in the first-order
fragment (uninterpreted function or predicate symbols must be avoided
to preserve decidability).
It is worth noticing that aggregation operators strongly depend on the
nature of each specific attribute; for example, natural aggregation
operators for memory and time are the maximum function and sum respectively.

\begin{example}[QoS specification]\label{ex:qosspec-detailed}
  With reference to \cref{ex:pop}, possible quantitative attributes of
  interest in an implementation of POP are
  $\constset[Q]=\{t, c, m\}$ representing CPU \textit{t}ime, monetary \textit{c}ost, 
  and \textit{m}emory usage, respectively.
  Then a QoS specification that characterises low computational
  costs, where no internal process consumes significant amount of
  resources, can be written according to \cref{def:qos-spec} as follows:
  \begin{align*}
    \Sigma = &\ \conf{\{0, 1\} \cup \{t, c, m\}, \{+, \cdot\} \cup
     \{\oplus^t, \oplus^c, \oplus^m\}, \{<\}} \\
   \Gamma = &\ \Gamma_\mathrm{RCF} \cup \Gamma^\prime_\mathrm{Low}
  \end{align*}%
  where $\oplus^t = +,\ \oplus^c = +,\ \oplus^m = \mathit{max}$, and
  $\Gamma^\prime_\mathrm{Low} = \{t \leq .01, c \leq .01, m \leq .01\}$
  \finex
\end{example}

From now on, we fix a set of constant symbols $\constset[Q]$ which we
omit from the set of \emph{QoS specifications}, that will be referred
to just as $\qosspecs$.  It is worth noting that, when $\constset[Q]$
is fixed, a QoS specification
$\conf{\Sigma, \Gamma_\mathrm{RCF} \cup \Gamma^\prime}$, is
  completely determined by $\Gamma^\prime$.
  Therefore, we can unabiguously refer
  to 
  a QoS specification using its set of formulas $\Gamma^\prime$.
  Thus, the QoS specification in \cref{ex:qosspec-detailed}
is $\Gamma^\prime_\mathrm{Low}$.
\begin{example}[QoS for POP]\label{ex:qosspec}
The following QoS specifications formalise the costs associated to different
activities in the POP protocol of \cref{ex:pop}.
  \begin{align*}
    & \Gamma^\prime_\mathrm{Chk} = \{ t \leq 5, c = 0.5, m = 0 \} \\
    & \Gamma^\prime_\mathrm{Mem} = \{ 1 \leq t \leq 6, c = 0, m \leq 64 \} \\
    & \Gamma^\prime_\mathrm{DB} = \{ t \leq 3 \implies (\exists x)(0.5 \leq x \leq 1 \land c = t \cdot x),
                                            t > 3 \implies c = 10,
                                            m \leq 5 \}
  \end{align*}
  Basically, $\Gamma^\prime_\mathrm{Chk}$
  formalizes the costs associated to the activity of
  integrity checking a message,
  $\Gamma^\prime_\mathrm{Mem}$ to the activity of a server
  receiving a message, and $\Gamma^\prime_\mathrm{DB}$ to
  establishing that the monetary cost is fixed if the insertion takes
  more than three time-units and it is a fraction of the
  execution time, otherwise.
  \finex
\end{example}
 
We now extend communicating systems (cf. \cref{sec:pre}) with
QoS-specifications.
\begin{definition}[QoS-extended CFSMs]\label{def:cfsm-extend}%
  A \emph{QoS-extended CFSM} is a tuple
  $\aCMq = \conf{\aCM, F, \qos}$ where:
  \begin{itemize}
  \item $\aCM = \conf{Q, q_0, \rightarrow }$ is a CFSM,
  \item $F \subseteq Q$ is a set of \emph{final states} of $\aCM$, and
  \item $\qos: Q \to \qosspecs$ maps states of $\aCM$ to QoS
	 specifications.
 \end{itemize}
 A \emph{QoS-extended communicating system} $S^{\mathit{QoS}}$ is a
 map $(\aCMq_{\p})_{\p \in \PSet}$ assigning an $\p$-local
 QoS-extended CFSM $\aCMq_{\p}$ to each $\p \in \PSet$.  A
 configuration $\csconf q \abuffer$ of $S^{\mathit{QoS}}$ is a
 \emph{final configuration} if $\vec q(\p) \in F_{\p}$ for every
 $\p \in \PSet$.
\end{definition}

\begin{example}[QoS-extended CFSMs] 
  \label{ex:pop-cfsm-qos}
  An extended CFSM of the POP client in \cref{ex:pop} with the QoS
  specifications of \cref{ex:qosspec} is as follows:
  %
  \begin{center}
  \begin{tikzpicture}[node distance=1cm and 1.3cm, every node/.style={scale=1,transform shape}, every label/.style={color=blue,fill=yellow!10,scale=.7}]
	 \def\low{\Gamma^\prime_\mathrm{Low}}
	 \def\db{\Gamma^\prime_\mathrm{DB}}
	 \def\mem{\Gamma^\prime_\mathrm{Mem}}
	 \def\chk{\Gamma^\prime_\mathrm{Chk}}
	 \node (d) at (0,0) {};
	 \node[cnode,right= .3cm of d, label={-95:$\low$}] (0) {};
	 \foreach \s/\l/\t in {0/\low/1,1/\db/2,2/\low/3,3/\db/4,4/\mem/5,5/\chk/6} {
		\node (\t) [cnode, right = of \s, label={$\l$}] {};
	 }
	 \foreach \s/\l/\t in {0/\aout[c][s][helo]/1,1/\ain[s][c][int]/2,2/\aout[c][s][read]/3,3/\ain[s][c][size]/4,4/\aout[c][s][retr]/5,5/\ain[s][c][msg]/6} {
		\path[line,sloped] (\s) edge node[below]{$\l$} (\t);
	 }
	 \node[cnode, above = of 0, label={95:$\low$}] (7) {};
	 \node[cnode, right = of 7, label={10:$\low$}, fill = black] (8) {};
	 \path[line] (d) -- (0);
	 \path[line] (6) edge[bend left=-35] node[below]{$\aout[c][s][ack]$} (2);
	 \path[line] (0) edge node[left]{$\aout[c][s][quit]$} (7);
	 \path[line] (2) edge[sloped] node[below,near end]{$\aout[c][s][quit]$} (7);
	 \path[line] (4) edge[sloped] node[above]{$\aout[c][s][quit]$} (7);
	 \path[line] (7) edge node[above]{$\ain[s][c][bye]$} (8);
  \end{tikzpicture}
  \end{center}
  where the filled state is the only final state.
  %
  %
  Each state is assigned a QoS specification given in
  \cref{ex:qosspec} according to the following idea.
  States where the client performs negligible computations are
  assigned the QoS specification
  $\Gamma^\prime_\mathrm{Low}$.
  The remaining states are assigned QoS specifications as follows.
  The local states where \p[c] performs a database insertion (right
  after receiving an $\msg[int]$ or $\msg[size]$ message) and those
  where \p[c] accesses the memory (right before receiving an unread
  e-mail) are constrained respectively by $\Gamma^\prime_\mathrm{DB}$
  and $\Gamma^\prime_\mathrm{Mem}$; finally,
   $\Gamma^\prime_\mathrm{Chk}$ constrains the states where \p[c]
	performs an integrity check (right after receiving an unread
	e-mail).
  %
  \finex
\end{example}

Notice that \cref{def:cfsm-extend} requires every
  state of a CFSM to be assigned a QoS specification. However, in most
  cases, most states will have the same QoS specification, as it is the
  case of $\Gamma^\prime_\mathrm{Low}$ in
\cref{ex:pop-cfsm-qos}; typically one only has to
  identify the QoS costs specific to few states.

The semantics of QoS-extended communicating systems is defined in the
same way as the semantics of communicating systems. This is a
consequence of the fact that QoS specifications do not have any effect
on communications.


%% file: logic.tex
%
%
To describe QoS properties we introduce $\QL$, a logical language akin
$\DLTL$.
%
%
\begin{definition}[QoS formulae]\label{def:globalQoSprop}
  The \emph{QoS logic} $\QL$ consists of the smallest set of
  formulae that can be obtained from the following grammar:
  \begin{align*}
    \gqosprop & \bnfdef \top \bnfmid \psi \bnfmid \neg \gqosprop \bnfmid \gqosprop \lor \gqosprop \bnfmid \gqosprop\ \mathcal{U}^{\mathsf{G}} \gqosprop
  \end{align*}
  where $\psi$ is a formula in a theory presentation in
  $\qosspecs$, and $\aG$ is a g-choreography over 
  $\PSet$ and $\msgset$ (see \cref{def:g-choreography}).
\end{definition}
Atomic formulae express constraints over quantitative attributes.
Akin $\DLTL$, properties of runs are linear temporal formulae
where the until operator is indexed with a global choreography
 $\aG$. In essence, the role of $\aG$ is to restrict the set of
  runs to be considered for the satisfiability of the until.
  Global choreographies are suitable for this purpose because they
  are a declarative and compact way of characterizing the behaviour 
  of asynchronous message-passing systems.
The possibility modality $\conf \aG \gqosprop$ is defined as
$\top \mathcal{U}^\aG \gqosprop$ and the necessity modality
$[\aG] \gqosprop$ is defined (dually) as
$\neg \conf \aG \neg \gqosprop$.
Finally, propositional
connectives $\land$ and $\implies$ are defined as usual.

The following example shows how to express non-functional properties
of specific runs of the system in $\QL$.
\begin{example}[QoS properties of POP protocol]\label{ex:qos-prop}
  %
  We can use the g-choreographies and the $\QL$ formula below to state
  that, unless the cost is zero for the first three e-mails read, the
  cost is bounded by $10$ times the CPU time, and the memory
  consumption is bounded by $5$.
  We define
  \begin{align*}
	 \gqosprop \equiv &\ [\aG_3]( c > 0 ) \implies [\aG_3;{\grec[][\aG_{\mathsf{msg}}]}]\big( ( c \leq t \cdot 10 ) \land (m \leq 5) \big)
							  \qqand[where]
	 \\	 
    \aG_3 = &\ \gint[][c][helo][s]; \gint[][s][int][c];\aG_{\mathsf{msg}};\aG_{\mathsf{msg}};\aG_{\mathsf{msg}} \qqand
	 \\
    \aG_{\mathsf{msg}} = &\ \gint[][c][read][s]; \gint[][s][size][c]; \gint[][c][retr][s]; \gint[][s][msg][c]; \gint[][s][ack][c]
  \end{align*}
  Intuitively, for $\gqosprop$ to hold, either the first three message
  retrievals must have zero cost in any run of the system, or on every
  subsequent message retrieval, the total cost and memory consumption
  fall within the specified bounds.
  \finex
\end{example}

A $\QL$ formula (like $\gqosprop$ in \cref{ex:qos-prop}) can be used
in quantitative analyses by \emph{aggregating} the values of the QoS
attributes along the runs of the system.
More precisely, given a run
$\pi$, our interpretation is that, for each transition $s_i \TRANSS{\al_i} s_{i+1}$ of $\pi$, the
obligations stated in the QoS specification of $s_i$ are met after
aggregating QoS information along $\pi$ from state $s_0$ up to state
$s_i$.
Therefore, a central notion in our framework is that of \emph{aggregation function}.
Given a QoS-extended communicating system
$\aCS$, an \emph{aggregation function} $\aggfn_{\aCS} : \Delta_{\aCS}
\to
\qosspecs$ yields a QoS specification capturing the cumulative QoS
attributes along a run $\pi
  \in
  \Delta_{\aCS}$ by \quo{summing-up} QoS specifications of participants' local states.
%
\begin{example}[Aggregation]\label{ex:aggregation}
  Recall the run $\pi$
  \eqref{eq:idea} from \cref{sec:intro}:
  \begin{align*}
	 \begin{tikzpicture}[node distance = .5cm and 3cm, every node/.style={font=\scriptsize}, every label/.style={color=blue,fill=yellow!10}]
		\node[label={above:$\{ c \leq 5, m = 0 \}$}] (aq0) {$\aQzero$};
		\node[left = 1cm of aq0] {\p:};
		\node[right = of aq0, label={above:$\{ 5 \leq c \leq 10,\ m < 3 \}$}] (aq1){$q_1$};
		\path[->,draw] (aq0) edge node[above] {$\aout$} (aq1);
		\node[below = of aq0, label={above: $\{c=0,m=0\}$}] (bq0){$\aQzero'$};
		\node[left = 1cm of bq0] {\q:};
		\node[right = of bq0, label={above:$\{10 \leq m \leq 50,\ c = 0.01 \cdot m \}$}] (bq1){$q_1'$};
		\path[->,draw] (bq0) edge node[above] {$\ain$} (bq1);
		\node(s0) [below = of bq0]{$\aConf_0$};
		\node[left = 1cm of s0] {$\pi$:};
		\node[right = of s0] (s1){$\aConf_1$};
		\node[right = of s1] (s2){$\aConf_2$};
		\path[->,draw] (s0) edge node[above] {$\aout$} (s1);
		\path[->,draw] (s1) edge node[above] {$\ain$} (s2);
	 \end{tikzpicture}
  \end{align*}


  %
  Let $c_{\p}^q$ (resp. $c_{\q}^q$) denote the value of the QoS
  attribute $c$ in the state $q$ of participant \p\ (resp. \q) and
  likewise for the attribute $\mathit{m}$.
  After $\pi$, we expect
	 $\max{\{m^{q_0}_{\p}, m^{q_1}_{\p}, m^{q'_0}_{\q}, m^{q'_1}_{\q}\}}$
   and
	 $c^{q_0}_{\p} + c^{q_1}_{\p} + c^{q'_0}_{\q} + c^{q'_1}_{\q}$
  to respectively be the memory consumption and the overall monetary 
  cost in $s_2$.
  This boils down to aggregate the QoS attributes $c$ and
  $m$ using the maximization and addition operations, respectively.
  \finex
\end{example}

Essentially, the aggregation in this case is obtained by
(1) instantiating the QoS specification associated to the local state
  of participants
  (this is done by renaming attributes as in \cref{ex:aggregation}); and
(2) adding an equation combining all the instances of QoS specifications.
The following formula captures this intuition
  and exemplifies one way in which the aggregation function could be defined.

%
%
%
\begin{example}[Aggregation]\label{ex:aggregation-function}
  Let $\aCS = (\conf{ \aCM_{\p}, \qos_{\p}})_{\p \in \PSet}$,
  we define the aggregation function 
  $\aggfn_{\aCS} : \Delta_{\aCS} \to \qosspecs$ to be
  $\aggfn_{\aCS} (\pi) = f(\pi) \cup g(\pi)$ where
  \begin{align*}
	 f(\pi) = & \bigcup\limits_{\substack{\p \in \PSet\\0 \leq i \leq n}} \psiinst[\qos_{\p}({{\vec q_i}(\p)})][@][{{\vec q_i}(\p)}]
	 \qqand
	 g(\epsilon) = 
		\left\{ \genericAttr = \left(\bigoplus^{\quad \genericAttr}\limits_{\substack{\p \in \PSet}} \psiinst[\genericAttr][@][{\vec q_0}(\p)]\right) \sst \genericAttr \in \constset[Q] \right\} 
		\\
   g(\pi) = &
		\left\{ \genericAttr = \left(\bigoplus^{\quad \genericAttr}\limits_{\substack{0 \leq i < n\\ \p = {\subject[\al_i]}}} \psiinst[\genericAttr][@][{\vec q_i}(\p)]\right) \oplus^\genericAttr \left(\bigoplus^{\quad \genericAttr}\limits_{\substack{\p \in \PSet}} \psiinst[\genericAttr][@][{\vec q_n}(\p)]\right) \sst \genericAttr \in \constset[Q] \right\}
		\qquad \text{if } \pi \not\eq \epsilon
  \end{align*}
where 
$\pi = \conf{\vec q_0 \, ; \, \vec b_0} \TRANSS{\al_0} \dots
  \TRANSS{\al_{n-1}} \conf{\vec q_n \, ; \, \vec b_n} \in
  \Delta_{\aCS}$,
  and
$\psiinst[\Pi] = \{\psiinst \sst \psi \in \Pi\}$ for a set of $\QL$ formulae $\Pi$, and $\psiinst$ is obtained by replacing each QoS attribute $c$
with the symbol $c_{\p}^q$ in the atomic formula $\psi$.
  The intuition is that $f(\pi)$ collects all the QoS specifications of the
  local states of the participants along the run $\pi$, and $g(\pi)$
  uses the aggregation operators to calculate the aggregated values of
  the QoS attributes in the run $\pi$. If we apply this aggregation
  function to the run $\pi$ in \cref{ex:aggregation}, we obtain the
  following:
  \begin{align*}
    f (\pi) =\ & \{ c^{\aQzero}_{\p} \leq 5, m^{\aQzero}_{\p} = 0 \} \cup \{ 5 \leq c^{q_1}_{\p} \leq 10,\ m^{q_1}_{\p} < 3 \}\\
    & \cup \{c^{\aQzero'}_{\q}=0,m^{\aQzero'}_{\q}=0\} \cup \{10 \leq m^{q_1'}_{\q} \leq 50,\ c^{q_1'}_{\q} = 0.01 \cdot m^{q_1'}_{\q} \} \\
    g (\pi) =\ & 	\left\{m = \max{\{m^{q_0}_{\p}, m^{q_1}_{\p}, m^{q'_0}_{\q}, m^{q'_1}_{\q}\}}, c = c^{q_0}_{\p} + c^{q_1}_{\p} + c^{q'_0}_{\q} + c^{q'_1}_{\q}\right\}
    \qquad\qquad\hfill\diamond
  \end{align*}
%
  %
%
\end{example}
%
It is important to emphasize that, in our conception, an aggregation
function relies on a run of the system as its input.  This run
inherently encompasses a specific sequential ordering of the actions
carried out by the participants.  The aggregation operators $max$ and
$+$ used in \cref{ex:aggregation} follow this interpration. As will
become clear in \cref{def:ql-semantics}, this interpretation is
sufficient for the purposes of this paper, since it enables $\QL$ to
specify temporal QoS properties about runs of the system.  However,
one might be interested in a different kind of aggregation that is
aware of local states that are executed in parallel. This may require
some care and possibly to exploit truly-concurrent models, such as
pomsets; this is left for future work.
%
 
The semantics of our logic is defined in terms of QoS-extended
communicating systems.
%
%

\begin{definition}[$\QL$ semantics]\label{def:ql-semantics}
  %
  Given a QoS-extended communicating system $\aCS$, an
  \emph{$\aCS$-model for a QoS property $\gqosprop$} is a pair
  $\conf{\pi, \pi' }$, where $\pi \in \Delta^\infty_{\aCS}$ contains
  a final configuration 
  (see \cref{def:cfsm-extend})
   and $\pi' \in \mathit{prf}(\pi)$ up to such configuration such that
  $\conf{\pi,\pi'} \models_{\aCS} \gqosprop$ where the 
  relation $\models_{\aCS}$ is defined as follows:
  $$
  \begin{array}{rcl}
    \conf{\pi, \pi'} \models_{\aCS} \gqosprop & \mbox{ iff } & \aggfn_{\aCS} (\pi') \vdash_{\mathit{RCF}} \gqosprop \ \ \text{ if } \gqosprop \text{ is an atomic formula}\\ 
    \conf{\pi, \pi'} \models_{\aCS} \neg \gqosprop & \mbox{ iff } & \conf{\pi, \pi'} \models_{\aCS} \gqosprop \text{ does not hold}\\ 
    \conf{\pi, \pi'} \models_{\aCS} \gqosprop_1 \lor \gqosprop_2 & \mbox{ iff } & \conf{\pi, \pi'} \models_{\aCS} \gqosprop_1 \mbox{ or } \conf{\pi, \pi'} \models_{\aCS} \gqosprop_2\\ 
    \conf{\pi, \pi'} \models_{\aCS} \gqosprop_1\ \mathcal{U}^\aG\ \gqosprop_2 & \mbox{ iff } & \text{there exists } \pi'' \text{such that } \rlang[\pi''] \in \hat{\rlang}[\aG], \\
    & & \pi'\pi'' \in \mathit{prf}(\pi) \text{ up to a final configuration in $\pi$,} \\
     & & \conf{\pi, \pi'\pi''} \models_{\aCS} \gqosprop_2 \text{ and, for all } \pi''' \in \mathit{prf} (\pi''),\\
     & & \text{if } \pi''' \neq \pi'' \text{ then } \conf{\pi, \pi'\pi'''} \models_{\aCS} \gqosprop_1.
  \end{array}
  $$
  \noindent

  A QoS property $\gqosprop$ is \emph{satisfiable} in $\aCS$ if there exists
  a run
  $\pi \in \Delta^\infty_{\aCS}$ such that
  $\conf{\pi, \epsilon} \models_{\aCS} \gqosprop$, and it is \emph{valid}
  (denoted as $\models_\aCS \gqosprop$) if,
   for all runs
  $\pi \in \Delta^\infty_{\aCS}$ that contain a final configuration,
  $\conf{\pi, \epsilon} \models_{\aCS} \gqosprop$.
\end{definition}
Negation and disjunction are handled in the standard way.
The definition of the until operator is similar to the standard
operator: $\gqosprop_2$ must hold at some point in the future,
i.e., $\pi'\pi''$ and $\gqosprop_1$ must hold up to that point; the
key difference is that the satisfaction of $\gqosprop_2$ is restricted
to runs where the extension $\pi''$ is in $\hat{\rlang}[\aG]$.
%
Finally, atomic formulae are handled by obtaining the aggregated QoS
of the accumulated run $\pi'$ and using the entailment relation of
RCFs.
%

%% file: analysis.tex

\SetKwFunction{qSat}{\textsc{qSat}}
\SetKwFunction{qModels}{\textsc{qModels}}
\SetKwFunction{qUntil}{\textsc{qUntil}}
\SetKwFunction{qRelease}{\textsc{qRelease}}
\SetKwFunction{traceOf}{\textsc{TraceOf}}

We now establish the semi-decidability of $\QL$ by presenting a
\mbox{$k$-bounded} semidecision procedure relying on three
algorithms: \qSat, \qModels, and \qUntil.
The \qSat algorithm is the main algorithm of the procedure and
determines whether a given formula is satisfiable in a given
system. It relies on \qModels to check if there is a run that
satisfies the formula which, in turn, uses \qUntil to handle the
$\mathcal{U}$ operator.
Let us start by looking at \qSat defined as:
%
%

\begin{algorithm}[H]
  \scriptsize
\DontPrintSemicolon
\SetAlgoLined
\SetKwProg{Fn}{}{\string:}{}
\Fn{\qSat{$\gqosprop,\aCS, k$}}
{
  $i = 0$\;
  \While{$i \leq k$}{
    \ForEach{$\pi \in \Delta_{\aCS}^i$}{
      \If{%
        the last configuration of $\pi$ is final and
        \qModels{$\gqosprop,\aCS,\pi,\epsilon$}
        \label{ln:qsat-qmodels}
      }{
        \Return{\textbf{true}}
      }
    }
    $i = i + 1$\;
  }
  \Return{\textbf{false}}
}
\end{algorithm}
\noindent
Basically, \qSat enumerates all the runs of $\aCS$ up to a given bound
$k$ and checks whether any of them satisfies $\gqosprop$ (recall
that $\Delta_{\aCS}^i$ is the set of all runs of $\aCS$ with length
$i$).
Let us now focus on the algorithm \qModels:
%

\begin{algorithm}[H]
  \scriptsize
\DontPrintSemicolon
\SetAlgoLined
\SetKwProg{Fn}{}{\string:}{}
\Fn{\qModels{$\gqosprop,\aCS,\pi,\pi'$}}
{
  \Switch{$\gqosprop$}{
  \Case{$\top$}{
    \Return{\textbf{true}}
    \label{ln:qsat-true}
  }
  \Case{$\psi$}{
    \Return{
      whether $\aggfn_{\aCS} (\pi') \vdash_{\mathit{RCF}} \psi$
    }
    \label{ln:qsat-rcf}
  }
  \Case{$\neg\gqosprop_1$}{
    \Return{\textbf{\emph{not}} \qModels{$\gqosprop_1,\aCS,\pi,\pi'$}}
  }
  \Case{$\gqosprop_1 \lor \gqosprop_2$}{
    \Return{\qModels{$\gqosprop_1,\aCS,\pi,\pi'$} \textbf{\emph{or}} \qModels{$\gqosprop_2,\aCS,\pi,\pi'$}}
    \label{ln:qsat-or}
  }
  \Case{$\gqosprop_1\ \mathcal{U}^\aG\ \gqosprop_2$}{
    \Return{\qUntil{$\gqosprop_1, \aG, \gqosprop_2,\aCS,\pi,\pi',\epsilon$}}
    \label{ln:qsat-until}
  }
  }
}
\end{algorithm}
\noindent
Following \cref{def:ql-semantics}, \qModels recursively inspects the
$\QL$ formula. It invokes 
\qUntil to handle the $\mathcal{U}$ operator
and the decision procedure of the theory of real-closed fields to
check the atomic formulae.
Let us now look at the algorithm \qUntil:
%

\begin{algorithm}[H]
  \scriptsize
\DontPrintSemicolon
\SetKwProg{Fn}{}{\string:}{}
\Fn{\qUntil{$\gqosprop_1, \aG, \gqosprop_2,\aCS,\pi,\pi', \pi''$}}
{
  \If{$\rlang[\pi''] \in \hat{\rlang}[\aG]$ and $\qModels{$\gqosprop_2,\aCS,\pi,\pi'\pi''$}$
    \label{ln:quntil-base-true-condition}
  }
  {
    \Return{\textbf{true}}
    \label{ln:quntil-base-true}
  }
  \ElseIf{not \qModels{$\gqosprop_1,\aCS,\pi,\pi'\pi''$}
    \label{ln:quntil-base-false-condition}
  }
  {
    \Return{\textbf{false}}
    \label{ln:quntil-base-false}
  }
  \Else{
    Let $\TRANSS{\al}\ q$ be the transition such that $\pi'\pi''\ \TRANSS{\al}\ q \in \mathit{prf}(\pi)$
    \label{ln:quntil-transition}
    \newline (takes the first transition in $\pi$ if $\pi'\pi'' = \epsilon$,
    \newline and it is not defined if $\pi'\pi'' = \pi$)\;
    \If{$\pi'\pi'' = \pi$ or $\rlang[\pi''\ \TRANSS{\al}\ q] \not\in \rlang[\aG]$}{
      \Return{\textbf{false}}
      \label{ln:quntil-empty-false}
    }
    \Else
    {
      \Return{\qUntil{$\gqosprop_1, \aG, \gqosprop_2,\aCS,\pi,\pi', \pi''\ \TRANSS{\al}\ q$}}
      \label{ln:quntil-recursive}
    }
  }
}
\end{algorithm}
\noindent
This procedure takes care of searching for a witness of the
existential in the semantics of $\mathcal{U}$ by starting in the
current prefix $\pi'$ and following the transitions of $\pi$.
According to \cref{def:ql-semantics}, \qUntil searches for a witness of the
existential part of $\mathcal{U}$.
It takes as parameters the complete run $\pi$, the prefix $\pi'$ at
which the $\mathcal{U}$ is being evaluated, and the current extension
$\pi''$ that is used to search for the witness.
If $\pi''$ is enough to reach a verdict, the algorithm returns true or
false accordingly (\cref{ln:quntil-base-true,ln:quntil-base-false}).
Otherwise, it tries to extend $\pi''$ by borrowing the next transition
of $\pi$ (\cref{ln:quntil-transition}).
If such extension exists and is a candidate for being in the language
of $\aG$, the algorithm recursively calls itself with the extended
prefix (\cref{ln:quntil-recursive}).
Hereafter, we fix a QoS-extended communicating system $\aCS$.

\begin{theorem}[\qSat is sound and $k$-bounded complete]
  \label{lem:qsat}
  Given a QoS formula $\gqosprop \in \QL$ and a bound $k$,
    \qSat{$\gqosprop, \aCS,k$} returns \textbf{true} iff
    there exists $\pi \in \Delta^i_{\aCS}$ such that
    $\conf{\pi, \epsilon} \models_{\aCS} \gqosprop$,
    for some $i \leq k$.
\end{theorem}
The soundness of \qSat immediately follows from the soundness of \qModels
(established in \cref{lem:qmodels} below) which, in turn, relies on the
soundness and completeness of \qUntil (cf.
\cref{lem:quntil-correct,lem:quntil-complete}, respectively).
This guarantees that the call to \qModels in \cref{ln:qsat-qmodels} of
\qSat returns true iff the run $\pi$ satisfies $\gqosprop$.
Note that \qSat is not guaranteed to be complete due to the bound
$k$.

\begin{lemma}[\qModels\ is sound and complete]
  \label{lem:qmodels}
  Given a QoS formula $\gqosprop \in \QL$ and runs $\pi, \pi' \in \Delta_{\aCS}$, 
  where $\pi' \in \mathit{prf}(\pi)$,
    \qModels{$\gqosprop, \aCS,\pi,\pi'$} returns \textbf{true} iff $\conf{\pi, \pi'} \models_{\aCS} \gqosprop$.
\end{lemma}

\begin{proof}
  By structural induction on $\gqosprop$.
  If $\gqosprop$ is $\top$, 
  the result follows trivially.
  If $\gqosprop$ is an atomic formula, the algorithm computes the
  aggregation over the run $\pi'$ (\cref{ln:qsat-rcf}) and invokes the
  decision procedure of RCFs to check whether $\aggfn_{\aCS} (\pi')$ entails
  $\gqosprop$ in the theory of real closed fields.
  If $\gqosprop$ is $\gqosprop_1 \lor \gqosprop_2$, the algorithm perform two recursive calls and returns true iff either $\conf{\pi, \pi'} \models_{\aCS} \gqosprop_1$ or $\conf{\pi, \pi'} \models_{\aCS} \gqosprop_2$. 
  If $\gqosprop$ is $\gqosprop_1\ \mathcal{U}^\aG\ \gqosprop_2$, the algorithm returns true iff \qUntil{$\gqosprop_1, \aG, \gqosprop_2,\aCS,\pi,\pi', \epsilon$} returns true. By \cref{lem:quntil-correct,lem:quntil-complete} this is equivalent to $\conf{\pi, \pi'} \models_{\aCS} \gqosprop_1\ \mathcal{U}^\aG\ \gqosprop_2$.
\qed
\end{proof}

We now prove the soundness and completeness of \qUntil.

\begin{lemma}[\qUntil\ is sound]
  \label{lem:quntil-correct}
  Given a QoS formula $\gqosprop_1, \gqosprop_2 \in \QL$, a
  g-choreography $\aG$, and runs $\pi, \pi', \pi''$ such that
  \begin{enumerate}[label = \alph*)] 
  \item
    \label{lem:quntil-correct:hyp-runs-in-sys}
    $\pi'\pi'' \in \mathit{prf}(\pi)$ and $\pi \in \Delta_{\aCS}$, and 
  \item
    \label{lem:quntil-correct:hyp-prefixes}
    for all $\pi''' \in \mathit{prf}(\pi'')$, if $\pi''' \neq \pi''$ then 
    $\conf{\pi, \pi'\pi'''} \models_{\aCS} \gqosprop_1$
  \end{enumerate}
  if \qUntil{$\gqosprop_1, \aG, \gqosprop_2,\aCS,\pi,\pi', \pi''$}
  returns \textbf{true}
  then
  $\conf{\pi, \pi'} \models_{\aCS} \gqosprop_1\ \mathcal{U}^\aG\
  \gqosprop_2$.
\end{lemma}
\begin{proof}
  The call to
  \qUntil{$\gqosprop_1, \aG, \gqosprop_2,\aCS,\pi,\pi', \pi''$} either
  reaches \cref{ln:quntil-base-true} or it reaches \cref{ln:quntil-recursive}
  and the recursive call returns true.  In the first case, we know
  $\rlang[\pi''] \in \hat{\rlang}[\aG]$ and that
  $\qModels{$\gqosprop_2,\aCS,\pi,\pi'\pi''$}$ returned true.  By
  \cref{lem:qmodels} it follows that
  $\conf{\pi, \pi'\pi''} \models_{\aCS} \gqosprop_2$.  Together with
  hypotheses \ref{lem:quntil-correct:hyp-runs-in-sys} and
  \ref{lem:quntil-correct:hyp-prefixes} the conditions of the
  semantics of the formula $\gqosprop_1\ \mathcal{U}^\aG\ \gqosprop_2$
  (see \cref{def:ql-semantics}) are met.  In the case of reaching
  \cref{ln:quntil-recursive} we know the recursive call
  \qUntil{$\gqosprop_1, \aG, \gqosprop_2,\aCS,\pi,\pi', \pi''\
	 \TRANSS{\al}\ q$} returned true.  
   Conditions
  \ref{lem:quntil-correct:hyp-runs-in-sys} and
  \ref{lem:quntil-correct:hyp-prefixes} applied to the input of the
  recursive call holds because of the way transition $\TRANSS{\al}\ q$
  was chosen and the fact that condition on
  \cref{ln:quntil-base-false-condition} returned false.  Therefore, we
  can take the output of the recursive call to satisfy
  \cref{lem:quntil-correct} as an inductive hypothesis and conclude
  $\conf{\pi, \pi'} \models_{\aCS} \gqosprop_1\ \mathcal{U}^\aG\
  \gqosprop_2$.
\qed
\end{proof}

\begin{lemma}[\qUntil is complete]
  \label{lem:quntil-complete}
  Given a QoS formula $\gqosprop_1, \gqosprop_2 \in \QL$, a g-choreography $\aG$, and runs $\pi, \pi', \pi''$ such that 
  \begin{enumerate}[label = \alph*)] 
    \item 
    \label{lem:quntil-complete:hyp-runs-in-sys}
    $\pi'\pi'' \in \mathit{prf}(\pi)$ and $\pi \in \Delta_{\aCS}$, and 
    \item 
    \label{lem:quntil-complete:hyp-prefixes}
    for all $\pi''' \in \mathit{prf}(\pi'')$, if $\pi''' \neq \pi''$ then 
    either $\rlang[\pi'''] \not\in \hat{\rlang}[\aG]$ or $\conf{\pi, \pi'\pi'''} \nvDash_{\aCS} \gqosprop_2$
  \end{enumerate}
    if \qUntil{$\gqosprop_1, \aG, \gqosprop_2,\aCS,\pi,\pi', \pi''$} returns \textbf{false}
    then $\conf{\pi, \pi'} \nvDash_{\aCS} \gqosprop_1\ \mathcal{U}^\aG\ \gqosprop_2$
\end{lemma}

\begin{proof}
  The call to \qUntil{$\gqosprop_1, \aG, \gqosprop_2,\aCS,\pi,\pi', \pi''$} reaches either \cref{ln:quntil-base-false}, \cref{ln:quntil-empty-false} or it reaches \cref{ln:quntil-recursive} and the recursive call returns false.
  In all cases, by condition \ref{lem:quntil-complete:hyp-prefixes} we know that no prefix of $\pi''$ could be witness of the existential in the semantics of $\gqosprop_1\ \mathcal{U}^\aG\ \gqosprop_2$ (see \cref{def:ql-semantics}) because it would need to both be in $\hat{\rlang}[\aG]$ and satisfy $\gqosprop_2$. Run $\pi''$ itself cannot be the witness for the same reasons due to the fact that condition in \cref{ln:quntil-base-true-condition} was not met. Which means either $\rlang[\pi''] \not\in \hat{\rlang}[\aG]$ or $\qModels{$\gqosprop_2,\aCS,\pi,\pi'\pi''$}$ returned false, therefore, using \cref{lem:qmodels}, either $\rlang[\pi''] \not\in \hat{\rlang}[\aG]$ or $\conf{\pi, \pi'\pi''} \nvDash_{\aCS} \gqosprop_2$.
  The only remaining possibility is for the witness to be a $\pi^\star$ such that $\pi'' \in \mathit{prf}(\pi^\star)$ and
  $\pi^\star \neq \pi''$.
  In the case of reaching \cref{ln:quntil-base-false}, we know that \qModels{$\gqosprop_1,\aCS,\pi,\pi'\pi''$} returned false. 
  By \cref{lem:qmodels} it follows that $\conf{\pi, \pi'\pi''} \nvDash_{\aCS} \gqosprop_1$.
  Therefore, extension $\pi^\star$ couldn't be a witness for the existential in the semantics of $\gqosprop_1\ \mathcal{U}^\aG\ \gqosprop_2$.
  In the case of reaching \cref{ln:quntil-empty-false}, candidate extension $\pi^\star$ does not exist or it is not in $\hat{\rlang}[\aG]$.
  In the case of reaching \cref{ln:quntil-recursive}, we know that \qUntil{$\gqosprop_1, \aG, \gqosprop_2,\aCS,\pi,\pi', \pi''\ \TRANSS{\al}\ q$} returned false.
  Notice that conditions
  \ref{lem:quntil-complete:hyp-runs-in-sys} and 
  \ref{lem:quntil-complete:hyp-prefixes} applied to the input of the recursive calls holds because
  of the way transition $\TRANSS{\al}\ q$ was chosen and
  that condition in \cref{ln:quntil-base-true-condition} was not met.
  Therefore, we can take the output of the recursive calls to satisfy \cref{lem:quntil-complete} as an inductive hypothesis and conclude that $\conf{\pi, \pi'} \nvDash \gqosprop_1\ \mathcal{U}^\aG\ \gqosprop_2$.
\qed
\end{proof}
Notice that the proof for \cref{lem:qmodels} uses \cref{lem:quntil-complete}
and \cref{lem:quntil-correct}, and that the proofs for \cref{lem:quntil-complete}
and \cref{lem:quntil-correct} use \cref{lem:qmodels}.
This does not undermine the soundness of the proofs because the lemmas are 
always (inductively) applied on smaller $\QL$ formulas.
Now that the soundness and completeness of \qModels and \qUntil is established, 
it remains to show their termination. Termination follows from the fact that
both the number of logical operators in $\gqosprop$ and the number of transitions in $\pi$ are finite.
The first guarantees \qModels eventually reaches a base case and the second guarantees \qUntil eventually 
reaches a base case.
Finally, the base case in \qModels, computing aggregation and checking entailment in the theory of real closed 
fields, terminates due to the decidability of RCFs \cite{tarski:RM-109}.

\subsection{A bounded model-checking approach for $\QL$}

Previous results allow for a straightforward bounded
model-checking approach for $\QL$.  Like for other model-checking
procedures for a language that admits negation, \qSat\ can be used to
check validity of a $\QL$ formula in a system $\aCS$ by checking the
satisfiability of the negated formula. This constitutes a
counterexample-finding procedure for $\QL$. The caveat is that
\qSat is a $k$-bounded semidecision procedure rather than a decision
procedure.  However, restricting to $\QL^-$, namely $\QL$ formulae
that do not contain the $\grecop$ operator in their choregraphies, we
 can find finite models of satisfiable formulae of $\QL^-$
(cf. \cref{lem:fmp}). 

\begin{theorem}[Finite model property of $\QL^-$]
  \label{lem:fmp}
  Given a QoS formula $\gqosprop \in \QL^-$, and runs $\pi \in \Delta_{\aCS}^\infty$, $\pi' \in \Delta_{\aCS}$ such that $\pi' \in \mathit{prf}(\pi)$.
  If $\conf{\pi, \pi'} \models_{\aCS} \gqosprop$ then there exists a finite run $\pi^- \in \Delta_{\aCS}$ such that $\pi^- \in \mathit{prf}(\pi)$ and $\conf{\pi^-, \pi'} \models_{\aCS} \gqosprop$.
\end{theorem}

\begin{proof}
  By structural induction on $\gqosprop$.
  If $\gqosprop$ is $\top$ or an atomic formula, take $\pi^- = \pi'$.
  If $\gqosprop$ is $\gqosprop_1 \lor \gqosprop_2$, we have that either $\conf{\pi, \pi'} \models_{\aCS} \gqosprop_1$ or $\conf{\pi, \pi'} \models_{\aCS} \gqosprop_2$. 
  By inductive hypothesis, either $\conf{\pi_1^-, \pi'} \models_{\aCS} \gqosprop_1$ or $\conf{\pi_2^-, \pi'} \models_{\aCS} \gqosprop_2$ for some finite $\pi_1^-,\pi_2^- \in \mathit{prf}(\pi)$. 
  Therefore, either $\conf{\pi_1^-, \pi'} \models_{\aCS} \gqosprop_1 \lor \gqosprop_2$ or $\conf{\pi_2^-, \pi'} \models_{\aCS} \gqosprop_1 \lor \gqosprop_2$.
  
  If $\gqosprop$ is $\gqosprop_1\ \mathcal{U}^\aG\ \gqosprop_2$, by \cref{def:ql-semantics} we have there exists $\pi''$ such that $\rlang[\pi''] \in \rlang[\aG]$,
   $\pi'\pi'' \in \mathit{prf}(\pi)$ up to a final configuration
   with $\conf{\pi, \pi'\pi''} \models_{\aCS} \gqosprop_2$, and for all $\pi''' \in \mathit{prf} (\pi'')$, if $\pi''' \neq \pi''$ then $\conf{\pi, \pi'\pi'''} \models_{\aCS} \gqosprop_1$.
  If we apply the inductive hypothesis to $\gqosprop_1$ and $\gqosprop_2$, we have
  there exists $\pi''$ such that $\rlang[\pi''] \in \rlang[\aG]$,
  $\pi'\pi'' \in \mathit{prf}(\pi)$ up to a final configuration
  with $\conf{\pi_2^-, \pi'\pi''} \models_{\aCS} \gqosprop_2$ for some 
  $\pi_2^- \in \Delta_{\aCS}$ such that 
  $\pi_2^- \in \mathit{prf}(\pi)$, and for all $\pi''' \in \mathit{prf} (\pi'')$, if $\pi''' \neq \pi''$ then  $\conf{\pi_1^-, \pi'\pi'''} \models_{\aCS} \gqosprop_1$ for some 
  $\pi_1^- \in \Delta_{\aCS}$ such that 
  $\pi_1^- \in \mathit{prf}(\pi)$. 
  Notice that since $\aG$ is $\grecop$-free, run $\pi''$ in the language $\rlang[\aG]$ is necessarily finite
  and so is the number of quantified runs $\pi'''$.
  Therefore, the number of runs $\pi_1^-$ involved in the previous statement is finite, and there is a maximum among their lengths,
  so we can take $\pi^-$ as the longest between $\pi_2^-$ and runs $\pi_1^-$.
  Since $\pi_2^-$ and all the $\pi_1^-$ are prefixes of $\pi$, then they will also be prefixes of $\pi^-$, and therefore we have the conditions to conclude $\conf{\pi^-, \pi'} \models_{\aCS} \gqosprop_1\ \mathcal{U}^\aG\ \gqosprop_2$.
  
  If the outermost operator in $\gqosprop$ is $\neg$, we need to consider al the possible cases for the immediate subformula of $\gqosprop$.
  If $\gqosprop$ is $\neg\psi$ with $\psi$ atomic formula, we have that $\conf{\pi, \pi'} \nvDash_{\aCS} \psi$. 
  Take $\pi^- \in \mathit{prf}(\pi)$ an extension of $\pi'$ whose last configuration is final.
  If $\gqosprop$ is $\neg(\gqosprop_1 \lor \gqosprop_2)$, we have that $\conf{\pi, \pi'} \nvDash_{\aCS} \gqosprop_1 \lor \gqosprop_2$. It follows that $\conf{\pi, \pi'} \nvDash_{\aCS} \gqosprop_1$ and $\conf{\pi, \pi'} \nvDash_{\aCS} \gqosprop_2$. By inductive hypothesis, there exists $\pi_1^- \in \Delta_{\aCS}$ such that $\pi_1^- \in \mathit{prf}(\pi)$ and $\conf{\pi_1^-, \pi'} \nvDash_{\aCS} \gqosprop_1$ and there exists $\pi_2^- \in \Delta_{\aCS}$ such that $\pi_2^- \in \mathit{prf}(\pi)$ and $\conf{\pi_2^-, \pi'} \nvDash_{\aCS} \gqosprop_2$. It is enough to take $\pi^-$ as the longest between $\pi_1^-$ and $\pi_2^-$.
  If $\gqosprop$ is $\neg(\gqosprop_1\ \mathcal{U}^\aG\ \gqosprop_2)$, we have that $\conf{\pi, \pi'} \nvDash_{\aCS} \gqosprop_1\ \mathcal{U}^\aG\ \gqosprop_2$. Therefore, for all $\pi''$ such that $\pi'\pi'' \in \mathit{prf}(\pi)$ up to a final configuration of $\pi$, if $\rlang[\pi''] \in \rlang[\aG]$, and $\conf{\pi, \pi'\pi''} \models_{\aCS} \gqosprop_2$, then there exists $\pi''' \in \mathit{prf} (\pi'')$, with $\pi''' \neq \pi''$ such that $\conf{\pi, \pi'\pi'''} \models_{\aCS} \neg\gqosprop_1$.
  If we apply the inductive hypothesis to $\gqosprop_2$ and $\neg\gqosprop_1$, we have that for all $\pi''$ such that $\pi'\pi'' \in \mathit{prf}(\pi)$ up to a final configuration of $\pi$, if $\rlang[\pi''] \in \rlang[\aG]$, and $\conf{\pi_2^- , \pi'\pi''} \models_{\aCS} \gqosprop_2$ for some $\pi_2^- \in \Delta_{\aCS}$ with $\pi_2^- \in \mathit{prf}(\pi)$, then there exists $\pi''' \in \mathit{prf} (\pi'')$, with $\pi''' \neq \pi''$ such that $\conf{\pi_1^-, \pi'\pi'''} \models_{\aCS} \neg\gqosprop_1$ for some $\pi_1^- \in \Delta_{\aCS}$ with $\pi_1^- \in \mathit{prf}(\pi)$. Notice that since $\aG$ is $\grecop$-free, any run in the language $\rlang[\aG]$ is necessarily finite. Therefore, there is a maximum among the lengths of the runs $\pi_2^-$, and we can take $\pi^-$ as the longest between $\pi_2^-$ and $\pi_1^-$.
\qed
\end{proof}

The proof of Thm.~\ref{lem:fmp} hints that the length of the 
run $\pi^-$ constitutes a suitable bound for \qSat;
which would turn \qSat into a decision procedure for $\QL^-$ 
if one could compute such bound.
Searching for counterexamples of an arbitrary formula $\gqosprop \in \QL$ up to a bounded number of unfoldings
of 
$\grecop$ 
is equivalent to searching for counterexamples in a formula $\hat \gqosprop$ in $\QL^-$ where each 
$\grecop$ 
has been replaced by a finite number of unfoldings. 
Which means that the bounded procedure for searching models of formulae in $\QL^-$ could be used to search for counterexamples of formulae in $\QL$.
Notice that \qSat can be easily extended to return the run that satisfies the formula, if there is one.
Such run can be used to identify the source of QoS formula violations when \qSat is used as a counterexample-finding procedure.

%% file: conclusions.tex
We presented a framework for the design and analysis of QoS-aware
distributed message-passing systems using choreographies and a general
model of QoS.
We tackle this problem by:
\begin{inparaenum}[1)]
\item abstractly representing QoS attributes as symbols denoting real
  values, whose behaviour is completely captured by a decidable RCFs
  theory,
\item extending the choreographic model of CFSM by associating QoS
  specifications to each state of the machine,
\item introducing $\QL$, a logic based on $\DLTL$, for
  expressing QoS properties with a straightforward satisfaction
  relation based on runs of communicating systems, and
\item giving a semi-decision procedure for $\QL$ and 
  defining a fragment $\QL^-$ that allowed us
  to give a bounded model-checking procedure for
  the full logic.
\end{inparaenum}
%
%
%
A prototype implementation of our procedure is under development.
It relies on the SMT solver Z3~\cite{demoura:tacas08} for the
satisfiability of the QoS constraints in atomic formulae and on
ChorGram~\cite{cgtCOORD20,chorgram} for the semantics of
g-choreographies and CFSMs.
%
%
%
An interesting by-product of our framework is that it could be used
for the monitoring of local computations to check at run-time if they
stay in the constraint of QoS specifications.
%
%
If static guarantees on QoS specifications are not possible, run-time
monitors can be easily attained by adapting techniques for monitor
generation from behavioural
types~\cite{FMT20,DBLP:journals/tcs/BocchiCDHY17}.

We identify two further main future research directions.
On the one hand, there is the theoretical question of whether $\QL$ is
decidable or not.
In this respect, the similarity of $\QL$ with $\DLTL$
(cf. \cref{sec:rw}) hints towards an affirmative answer suggesting
that the problem can be translated to checking emptyness of B\"{u}chi
automata~\cite{buchi:iclmps62} corresponding to $\QL$ formulae.
However, the decidability of QL is not so easy to attain. 
In general, a communicating system may yield an infinite state space 
due to many reasons so satisfaction might not be possible in a finite 
number of steps. For instance, due to potentially infinite instantiations 
of QoS attributes or that no final configuration might be reachable.
On the other hand, the usability of the framework could be improved
through two extensions of $\QL$ and a less demanding way 
of modeling QoS-extended communicating systems.
The first extension of $\QL$ are \emph{selective} aggregation, enabling the
aggregation of QoS attributes only for some specific states of runs.
This can be done by extending the grammar of g-choreographies given in
\cref{def:g-choreography} with an extra production of the shape
$\aG \bnfdef \cdots \bnfmid \ggrp$, \quo{bracketing} the part of the
choreography relevant for the aggregation. Notice that the run still has to match the whole choreography.
A second extension of $\QL$ is the introduction of \emph{wildcards} as a
mechanism to \quo{ignore} a subchoreography.
Syntactically, it can be represented by, once again, extending the
grammar given in \cref{def:g-choreography} with an extra production,
with shape $\aG \bnfdef \cdots \bnfmid \_\ $, where $\_$ is interpreted
as a wildcard and plays the role of matching any possible
g-choreography. In this case, the shape of the part of the run that matches the wildcard 
is disregarded but attributes are aggregated along the whole run.
Finally, a less demanding way of modeling QoS-extended systems could be achieved 
by extending g-choreographies with QoS specifications annotating specific
interactions and extending the projection of g-choregraphies into CFSMs taking into account
such annotations.


%% file: main.bbl
\begin{thebibliography}{10}

\bibitem{bpmn}
{O}bject~{M}anagement {G}roup:
\newblock {{B}usiness {P}rocess {M}odel and {N}otation}
  {\url{http://www.bpmn.org}}.

\bibitem{bon18}
Bon\'er, J.:
\newblock {Reactive Microsystems - The Evolution Of Microservices At Scale}.
\newblock O'Reilly (2018)

\bibitem{fmmt20}
Frittelli, L., Maldonado, F., Melgratti, H.C., Tuosto, E.:
\newblock A choreography-driven approach to apis: The opendxl case study.
\newblock In \cite{b20}, 107--124.

\bibitem{b20}
Bliudze, S., Bocchi, L.: 
\newblock Proc. of Coordination Models and Languages -
  22nd Intl. Conf. {IFIP} {WG} 6.1, Valletta, Malta, June 15-19, 2020.
  Vol. 12134 of LNCS, Springer (2020).
  
\bibitem{DBLP:journals/software/AutiliIT15}
Autili, M., Inverardi, P., Tivoli, M.:
\newblock Automated synthesis of service choreographies.
\newblock {IEEE} Softw. \textbf{32}(1) (2015)  50--57.

\bibitem{w3c:wsdl20}
{W}orld {W}ide~{W}eb {C}onsortium:
\newblock Web services description language (wsdl) version 2.0 part 1: Core
  language.
\newblock On-line Available at \url{https://www.w3.org/TR/wsdl20/}.

\bibitem{ich12}
Ivanovi{\'{c}}, D., Carro, M., Hermenegildo, M.V.:
\newblock A constraint-based approach to quality assurance in service
  choreographies.
\newblock In Liu, C., et.al. eds.: Service-Oriented
  Computing, Berlin, Heidelberg, Springer Berlin Heidelberg (2012)  252--267.

\bibitem{kgi13}
Kattepur, A., Georgantas, N., Issarny, V.:
\newblock Qos analysis in heterogeneous choreography interactions.
\newblock In Basu, S., et.al. eds.: Service-Oriented
  Computing, Berlin, Heidelberg, Springer Berlin Heidelberg (2013)  23--38.

\bibitem{gpsy16}
G{\"{u}}demann, M., Poizat, P., Sala{\"{u}}n, G., Ye, L.:
\newblock Verchor: {A} framework for the design and verification of
  choreographies.
\newblock {IEEE} Trans. Serv. Comput. \textbf{9}(4) (2016)  647--660.

\bibitem{adsgpt19}
Autili, M., Salle, A.D., Gallo, F., Pompilio, C., Tivoli, M.:
\newblock Chorevolution: Automating the realization of highly-collaborative
  distributed applications.
\newblock In: Proc. of Coordination Models and Languages - 21st Intl. Conf. IFIP WG 6.1, Kongens Lyngby, Denmark, June 17-21, 2019,
  Springer (2019)  92--108.

\bibitem{bmt19}
Bocchi, L., Melgratti, H.C., Tuosto, E.:
\newblock On resolving non-determinism in choreographies.
\newblock Log. Methods Comput. Sci. \textbf{16}(3) (2020).

\bibitem{DBLP:conf/popl/BasuBO12}
Basu, S., Bultan, T., Ouederni, M.:
\newblock Deciding choreography realizability.
\newblock In Field, J., et.al. eds.: Proc. of the 39th {ACM}
  {SIGPLAN-SIGACT} Symp. on Principles of Programming Languages, {POPL}
  2012, Philadelphia, Pennsylvania, USA, January 22-28, 2012, {ACM} (2012)
  191--202.

\bibitem{hlvccdmprt16}
H{\"{u}}ttel, H., Lanese, I., Vasconcelos, V.T., Caires, L., Carbone, M.,
  Deni{\'{e}}lou, P., Mostrous, D., Padovani, L., Ravara, A., Tuosto, E.,
  Vieira, H.T., Zavattaro, G.:
\newblock Foundations of session types and behavioural contracts.
\newblock {ACM} Comput. Surv. \textbf{49}(1) (2016)  3:1--3:36.

\bibitem{brand:jacm-30_2}
Brand, D., Zafiropulo, P.:
\newblock On communicating finite-state machines.
\newblock Jour. of the {ACM} \textbf{30}(2) (1983)  323--342.

\bibitem{henriksen:apal-96_1_3}
Henriksen, J.G., Thiagarajan, P.:
\newblock Dynamic linear time temporal logic.
\newblock Annals of Pure and Applied Logic \textbf{96}(1--3) (1999)  187--207.

\bibitem{tuosto:jlamp-95}
Tuosto, E., Guanciale, R.:
\newblock Semantics of global view of choreographies.
\newblock Jour. of Logical and Algebraic Methods in Programming \textbf{95}
  (2018)  17--40.

\bibitem{tarski:RM-109}
Tarski, A.:
\newblock A decision method for elementary algebra and geometry.
\newblock Memorandum RM-109, {RAND} Corporation (1951).

\bibitem{Aleti2013SoftwareAO}
Aleti, A., Buhnova, B., Grunske, L., Koziolek, A., Meedeniya, I.:
\newblock Software architecture optimization methods: A systematic literature
  review.
\newblock IEEE Transactions on Software Engineering \textbf{39} (2013)
  658--683.

\bibitem{hayyolalam:jnca-110}
Hayyolalam, V., Pourhaji~Kazem, A.A.:
\newblock A systematic literature review on {{QoS-aware}} service composition
  and selection in cloud environment.
\newblock Jour. of Network and Comp. Applications \textbf{110} (Timothy
  {C.} May 2018)  52--74.

\bibitem{gdgln16}
Giachino, E., de~Gouw, S., Laneve, C., Nobakht, B.:
\newblock Statically and dynamically verifiable {SLA} metrics.
\newblock In {\'{A}}brah{\'{a}}m, E., et.al. eds.:
  Theory and Practice of Formal Methods - Essays Dedicated to Frank de Boer on
  the Occasion of His 60th Birthday. Vol. 9660 of LNCS, Springer (2016)  211--225.

\bibitem{Rosenthal90}
Rosenthal, K.:
\newblock Quantales and Their Applications. Vol. 234 of Pitman Research Notes
  in Mathematics Series.
\newblock Longman Scientific \& Technical (1990).

\bibitem{buscemi:esop07}
Buscemi, M.G., Montanari, U.:
\newblock Cc-pi: A constraint-based language for specifying service level
  agreements.
\newblock In DeNicola, R., ed.: Proc. of 16th European Symp. on
  Programming, ESOP 2007. Vol. 4421 of LNCS, Springer (2007)  18--32.

\bibitem{lm05}
Lluch-Lafuente, A., Montanari, U.:
\newblock Quantitative $\mu$-calculus and {CTL} based on constraint semirings.
\newblock Elec. Notes in Theo. Comp. Sci. \textbf{112} (2005)
  37--59 Proc. of Quantitative Aspects of Programming Languages (QAPL 2004).

\bibitem{dfmpt05}
{De Nicola}, R., Ferrari, G., Montanari, U., Pugliese, R., Tuosto, E.:
\newblock A process calculus for {QoS}-aware applications.
\newblock In Jacquet, J., et.al. eds.: Proc. of Coordination Models and Languages, Namur, Belgium, April
  20-23, 2005. Vol. 3454 of LNCS, Springer (2005)  33--48.

\bibitem{martinezsune:coordination19}
{Martinez Su\~{n}\'{e}}, A.E., {Lopez Pombo}, C.G.:
\newblock Automatic quality-of-service evaluation in service-oriented
  computing.
\newblock In Nielson, H.R., et.al. eds.: Proc. of Coordination
  Models and Languages - 21st Intl. Conf. {IFIP} {WG} 6.1. Vol. 11533
  of LNCS, Springer (June 2019)  221--236.

\bibitem{bhty10}
Bocchi, L., Honda, K., Tuosto, E., Yoshida, N.:
\newblock A theory of design-by-contract for distributed multiparty
  interactions.
\newblock In Gastin, P.,et.al. eds.: Proc. of 21th Intl. Conf. Concurrency Theory (CONCUR), Paris, France, August 31-September 3, 2010. 
Vol. 6269 of LNCS, Springer (2010)  162--176.

\bibitem{vissani:places15}
Vissani, I., {Lopez Pombo}, C.G., Tuosto, E.:
\newblock Communicating machines as a dynamic binding mechanism of services.
\newblock In Gay, D., et.al. eds.: Proc. of Programming Language Approaches to Concurrency- and
  Communication-cEntric Software (PLACES). Vol. 203 of Elect.
  Proc. in Theo. Comp. Sci.. (April 2016)  85--98.


\bibitem{pnueli:tcs-13_1}
Pnueli, A.:
\newblock The temporal semantics of concurrent programs.
\newblock Theo. Comp. Sci. \textbf{13}(1) (1981)  45--60.

\bibitem{pratt:ieee-sfcs76}
Pratt, V.R.:
\newblock Semantical consideration on floyd-hoare logic.
\newblock In Carlyle, et.al. eds.: Proc. of 17th.
  Annual Symp. on Foundations of Comp. Sci. -- SFCS 1976, {IEEE}
  Computer Society (1976)  109--121.


\bibitem{vardi:eptcs-54}
Vardi, M.Y.:
\newblock The {{Rise}} and {{Fall}} of {{LTL}}: {{Invited Presentation}} at the
  {{2nd. Intl. Symp.}} on {{Games}}, {{Automata}}, {{Logics}} and
  {{Formal Verification}}.
\newblock Elec. Proc. in Theo. Comp. Sci. \textbf{54} (2011).

\bibitem{degiacomo:ijcai13}
De~Giacomo, G., Vardi, M.Y.:
\newblock Linear temporal logic and linear dynamic logic on finite traces.
\newblock In: Proc. of 23rd Intl. Conf. on Artificial Intelligence (IJCAI), {Beijing, China},
  {AAAI Press} (2013)  854--860.

\bibitem{rfc937}
\newblock {Post Office Protocol: Version 2}.
\newblock RFC 937 (1985).





\bibitem{demoura:tacas08}
{de Moura}, L.M., Bj{\o}rner, N.:
\newblock {Z3:} an efficient {SMT} solver.
\newblock In Ramakrishnan, et.al. eds.: Proc. of 14th
  Intl. Conf. Tools and Algorithms for the Construction and
  Analysis of Systems (TACAS). Vol. 4963 of
  LNCS, Springer (2008)  337--340.

\bibitem{cgtCOORD20}
Coto, A., Guanciale, R., Tuosto, E.:
\newblock Choreographic development of message-passing applications - {A}
  tutorial.
\newblock \newblock In \cite{b20}, 20--36.

\bibitem{chorgram}
Coto, A., Guanciale, R., Lange, J., Tuosto, E.:
\newblock \chorgram: tool support for choreographic deveelopment.
\newblock Available at {\chorgramsite} (2015).

\bibitem{FMT20}
Francalanza, A., Mezzina, C.A., Tuosto, E.:
\newblock Towards choreographic-based monitoring.
\newblock In: Reversible Computation: Extending Horizons of Computing -
  Selected Results of the {COST} Action {IC1405}. Vol. 12070 of LNCS,
\newblock Springer (2020)  128--150.

\bibitem{DBLP:journals/tcs/BocchiCDHY17}
Bocchi, L., Chen, T., Demangeon, R., Honda, K., Yoshida, N.:
\newblock Monitoring networks through multiparty session types.
\newblock Theor. Comput. Sci. \textbf{669} (2017)  33--58.

\bibitem{buchi:iclmps62}
B{\"{u}}chi, J.R.:
\newblock On a decision method in restricted second order arithmetic.
\newblock In: Proc. of the Intl. Congress on Logic, Method, and
  Philosophy of Science, Stanford, CA, USA, Stanford University, Stanford
  University Press (1962)  1--11.



\bibitem{biere21}
Biere, A., Heule, M., Van~Maaren, H., Walsh, T., eds.:
\newblock Handbook of {{Satisfiability}}: {{Second Edition}}. 
\newblock {IOS Press} (February 2021)

\end{thebibliography}
